\documentclass[letterpaper, 11 pt]{article} 
\usepackage{amsmath,amsfonts,amssymb,color}
\usepackage{mathtools}
\usepackage{xspace}
\usepackage{dsfont}
\usepackage[dvipsnames]{xcolor}
\definecolor{mygreen}{rgb}{0.0, 0.5, 0.0}
\definecolor{winered}{rgb}{0.8,0,0}
\definecolor{myblue}{rgb}{0,0,0.8}
\usepackage{tikz}
\usepackage{wrapfig}
\usepackage{amsthm}
\usepackage{empheq}

\usepackage{amsthm}
\newtheorem{Problem}{Problem}

\newtheorem{definition}{Definition}
\newtheorem{theorem}{Theorem}
\newtheorem{lemma}{Lemma}

\newtheorem{corollary}{Corollary}
\newtheorem{remark}{Remark}

\DeclareMathOperator*{\argmax}{\arg\!\max}
\DeclareMathOperator*{\argmin}{\arg\!\min}
\DeclarePairedDelimiterX{\norm}[1]{\lVert}{\rVert}{#1}

\usepackage{graphicx}
\usepackage{caption}
\usepackage{subcaption}
\usepackage{algorithm}% http://ctan.org/pkg/algorithms
\usepackage{algpseudocode}% 
\usepackage{hyperref}
\usepackage{epsfig}
\usepackage{array}
\usepackage{multirow}
\usepackage{epstopdf}
\usepackage{tikz}
\usepackage{relsize}
\usetikzlibrary{shapes,arrows}
\hypersetup{
    colorlinks=true,
    linkcolor={winered},
    citecolor={mygreen}
}
\usepackage{microtype}
\usepackage{graphicx}
\usepackage{subcaption}
\usepackage{cite}

\newcommand{\scenario}[1]
{{\fontsize{8.5}{8.5}\selectfont\sf #1}\xspace}
\def\BibTeX{{\rm B\kern-.05
em{\sc i\kern-.025em b}\kern-.08em
    T\kern-.1667em\lower.7ex\hbox{E}\kern-.125emX}}
\usepackage[margin=1in]{geometry}

\title{A Unified Framework for Joint Sensor Placement
and Scheduling for Intrusion Detection}
\author{Jayanth Bhargav, Mahsa Ghasemi and Shreyas Sundaram}
\date{}

\begin{document}
\emergencystretch=3em
\maketitle
\newcommand{\sense}{$\scenario{SensorPlacement}$}
\newcommand{\orient}{$\scenario{OrientationScheduling}$}
\newcommand{\abs}[1]{\left| #1 \right|}
\newcommand{\cals}{\mathcal{S}}
\newcommand{\sopt}{\mathcal{S}^{\scenario{OPT}}}

\begin{abstract}
We consider an intrusion detection task in which a defender must jointly optimize sensor placement locations and orientations to minimize the probability of missed detection of an intruder traversing a protected environment. We decompose this problem into a meta problem, termed SensorPlacement, and an embedded subproblem, termed OrientationScheduling. The OrientationScheduling subproblem, for a fixed sensor placement, is modeled as a 2-player zero-sum game between the defender and the intruder, where the defender seeks an orientation strategy for the deployed sensors to minimize the probability of missed detection, while the intruder seeks a path selection strategy to maximize it. Since the defender’s strategy space grows combinatorially with the number of sensors and orientations, solving the game via standard linear programming becomes prohibitive. To this end, we develop an iterative and efficient equilibrium-seeking algorithm that exploits the structure of the game’s payoff function and establish theoretical guarantees for convergence to the Nash equilibrium (NE) of the game. This NE value of the game is then used as a utility measure in the  SensorPlacement meta problem. We show that this game-value-based utility function is weakly submodular over the set of sensor placements and propose a greedy placement algorithm with near-optimality guarantees. To our knowledge, this is the first unified framework to integrate game-theoretic utility design with (weak) submodular optimization, enabling principled joint optimization of sensor placement and orientation scheduling. With extensive simulations, we demonstrate that the proposed approach achieves near-optimal detection performance while significantly reducing computation time compared to  baselines.
\end{abstract}

\section{Introduction}
In modern security and surveillance systems, effective monitoring depends on both where sensors are placed and how their sensing modes (i.e., field-of-view (FoV) orientations) are scheduled over time \cite{guvensan2011coverage,murray2007coverage}. While prior work has extensively studied sensor placement \cite{dhillon2003sensor,bhargav2023complexity,jourdan2008optimal} and orientation scheduling \cite{bhargav2025sensor, xu2024performance, osais2009sensor} as separate problems, these two decisions are fundamentally coupled in practice. Sensor placement determines the geometric structure of coverage, e.g., the overlap between sensing regions, the set of intruder paths that can potentially be monitored. Scheduling determines how this spatial coverage is realized over time, particularly when sensors have steerable orientations with limited FoV and operational constraints.

\begin{figure}[!t]
\centering
\includegraphics[width=360pt]{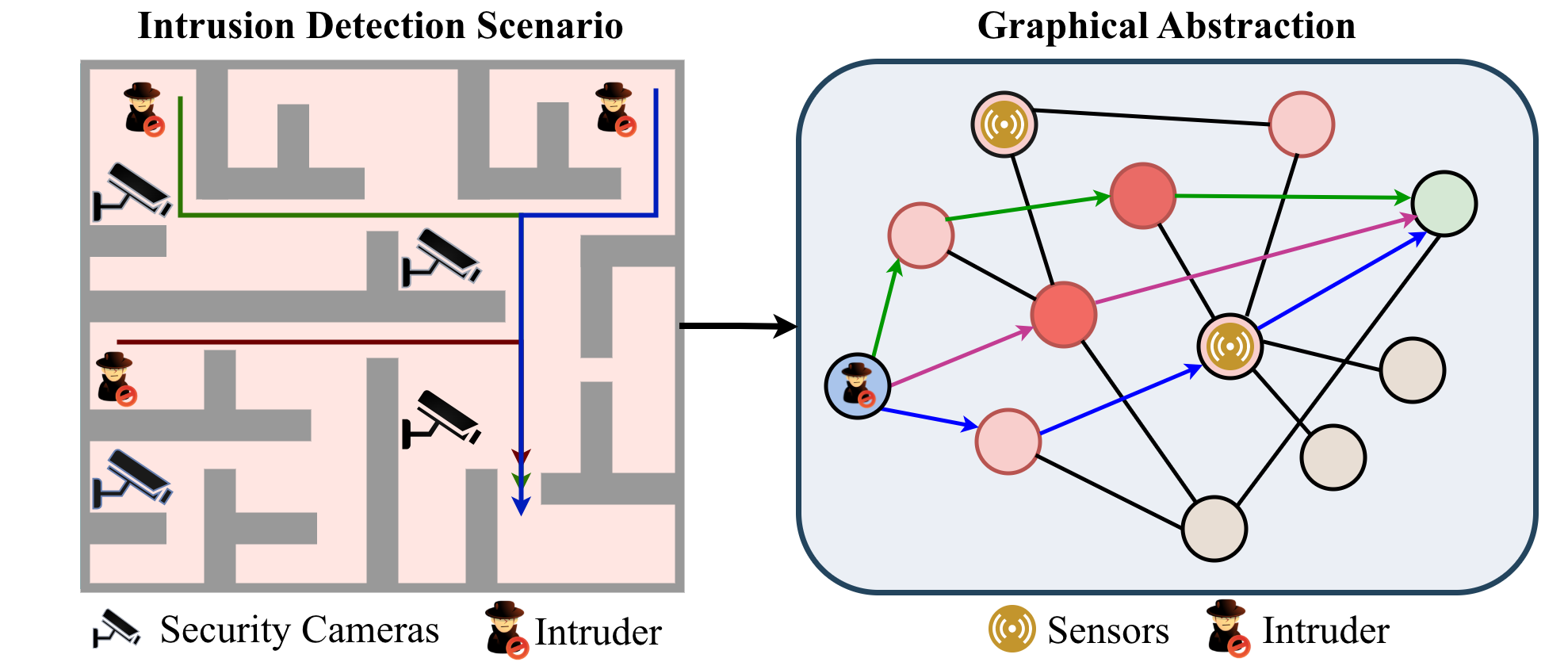}
    \caption{Illustration of an intrusion detection scenario and its graph abstraction. Left: Physical environment consisting of multiple regions, showing candidate sensor placements (security cameras) and representative intruder trajectories (colored paths). Right: Graph representation of the environment, where nodes correspond to different regions and edges encode feasible transitions between regions. Colored edges depict representative intruder paths encoded within the graph structure.}
    \label{fig:game}
\end{figure}

Optimizing only placement is insufficient if sensing orientations are not coordinated appropriately. Even a well-distributed deployment may exhibit temporal blind spots if scheduling is poorly designed~\cite{murray2007coverage}. Conversely, sophisticated sensor scheduling strategies cannot compensate for fundamentally weak spatial placement~\cite{dhillon2003sensor,jourdan2008optimal}. The placement configuration defines the feasible coverage regions of the scheduler, while scheduling dictates how effectively the deployed sensing resource is utilized. This intrinsic coupling motivates a joint optimization framework that accounts for both spatial and temporal decisions. Also, the need for joint design becomes more pronounced in adversarial environments. In the presence of an adaptive intruder, deterministic sensing strategies are vulnerable to exploitation \cite{an2017stackelberg}. An intelligent adversary can identify predictable orientation patterns or recurring coverage gaps and choose traversal paths accordingly. Game-theoretic models provide a principled framework for capturing this strategic interaction between the defender and the intruder \cite{pirani2021game}. Within such frameworks, the defender employs randomized scheduling policies, typically represented as mixed strategies, to reduce predictability and ensure robustness against strategic evasion.

In this paper, we consider an intrusion detection task in which a defender must deploy a limited number of sensors to monitor a protected environment against an adaptive intruder. Each sensor can be placed at selected candidate locations and may operate under multiple orientations (or sensing modes). 
%The defender seeks to minimize the probability of missed detection, while the intruder strategically tries to select traversal paths through the environment to evade sensing coverage, i.e.,  maximize the probability of missed detection. 
The key challenge lies in the coupled nature of spatial sensor deployment and orientation scheduling decisions, particularly under strategic adversaries.

\subsection{Contributions}

We propose a unified framework that jointly optimizes both the sensor placement and orientation scheduling in the presence of an adaptive intruder. We decompose the joint optimization task into a meta problem, termed \sense, and an embedded subproblem, termed \orient, enabling a principled integration of spatial deployment and temporal strategy design. For a given sensor placement, we model the \orient~subproblem as a 2-player zero-sum game between the defender and the intruder over a graph (see Figure~\ref{fig:game}). The defender seeks a randomized orientation strategy to minimize the probability of missed detection, while the intruder seeks a randomized path selection strategy to maximize it. The NE value of this game characterizes the equilibrium payoff under mutual best responses, i.e., the best expected missed detection probability that the defender can guarantee against an intruder that optimally responds. We therefore use the NE value as the utility function for the meta $\scenario{SensorPlacement}$ problem, which seeks to identify the placement that optimizes this detection performance utility.

Our key algorithmic and theoretical contributions are summarized as follows.

\begin{itemize}

\item \textit{Scalable algorithm for equilibrium-seeking --}
In the \orient~game, the defender’s strategy space grows combinatorially with the number of sensors and orientations, and thus solving for the game value using standard linear programming (LP) approaches becomes computationally prohibitive. We exploit structural properties of the game's payoff function and  develop an efficient equilibrium-seeking algorithm that is based on a multiplicative weight update procedure (as in \cite{freund1996game}), and establish theoretical guarantees for convergence to the NE value of the sensor scheduling game.

\item \textit{Efficient algorithm for sensor placement --}
We prove that the game-value-based utility function of the \sense~problem is weakly submodular over the set of  sensors. Leveraging this property, we develop a greedy placement algorithm and establish near-optimality guarantees for the sensor placement problem.

\end{itemize}

Finally, through extensive simulations, we demonstrate that our approach achieves near-optimal detection performance while reducing computational complexity compared to existing baselines.

The game-theoretic sensor scheduling problem (i.e., \orient) was first introduced in our prior conference paper~\cite{bhargav2025sensor}. In~\cite{bhargav2025sensor}, we developed a distributed variant of the weighted majority algorithm~\cite{freund1996game} and established convergence guarantees to the NE of the scheduling game. The work in~\cite{bhargav2025sensor} also focused on online refinement of equilibrium strategies under uncertain sensor models. This paper makes two significant advances. First, we develop an equilibrium-seeking algorithm with faster convergence rates and better computational efficiency compared to~\cite{bhargav2025sensor}. Second, we extend the scope of the problem from orientation scheduling alone to the joint placement and scheduling setting by embedding the game-theoretic orientation scheduling framework within a meta-level spatial placement optimization. To the best of our knowledge, this is the first unified framework to integrate game-theoretic utility design with (weak) submodular optimization, enabling principled joint optimization of sensor placement and orientation scheduling.

\subsection{Related Work}

In this section, we review existing works on game-theoretic and submodular-optimization-based approaches for sensor resource management, and discuss how our contributions relate to and extend beyond these works.

\noindent \textbf{Game-Theoretic Sensor Resource Management --}
Security resource allocation problems have been extensively studied using game-theoretic models, particularly within the Stackelberg security game framework \cite{sinha2018stackelberg,korzhyk2010complexity}. In such settings, the defender commits to a mixed strategy, and the attacker best-responds after observing it. For example, \cite{paruchuri2008playing} models airport patrol scheduling as a Stackelberg game where the defender allocates patrol resources and the attacker chooses a pure attack strategy. These models are well-suited for settings where the defender can commit to a strategy and the attacker responds sequentially. In contrast, our work considers a concurrent-move setting in which both the defender and the intruder select mixed strategies without real-time knowledge of the opponent’s action. Under such conditions, neither player can observe or respond to the exact realized strategy of the other. As shown in \cite{korzhyk2011stackelberg}, Stackelberg and Nash equilibria coincide in this class of games. We therefore adopt a zero-sum game formulation to model sensor orientation scheduling, which naturally captures mutual uncertainty and strategic adaptation. \\

\noindent \textbf{Solving Large Games with Combinatorial Action Spaces --}
A significant body of work addresses the challenge of computing equilibria in games with large or combinatorial strategy spaces \cite{ganzfried2015endgame,lipton2003playing,sandholm2015abstraction}. Techniques such as abstraction, sampling, and iterative regret minimization have been developed to enable scalable equilibrium computation. For instance, \cite{li2020structure} proposes an iterative structure-learning framework to approximate solutions in many-player games. In \cite{li2023decision,brown2019deep}, counterfactual regret minimization methods like \scenario{CRF-MIX} are developed to solve team-adversary games with combinatorial action spaces. However, these works do not provide theoretical guarantees for near-optimality of the resulting strategies to the game's equilibrium. Unlike generic abstraction or regret-minimization frameworks, our approach is tailored to the sensing structure of the problem and is equipped with convergence guarantees to the exact NE.\\

\noindent \textbf{Submodular Optimization for Sensor Placement --}
A substantial body of literature addresses sensor placement by leveraging submodularity or weak-submodularity properties of the underlying utility functions~\cite{hashemi2020randomized,ye2021near,ghasemi2019online,bhargav2024submodular}. These structural properties enable efficient approximation via greedy algorithms with provable near-optimality guarantees~\cite{nemhauser1978analysis,krause2007near,kaya2025randomized,hashemi2019submodular,bhargav2025robust}. In line with this direction, we establish that the game-value-based objective function of the \sense~problem is weakly submodular over the set of deployed sensors. This structural characterization allows us to design a greedy placement algorithm with theoretical performance guarantees. To the best of our knowledge, integrating game-value-based utilities within a weak-submodular optimization framework has not been previously explored in the sensor selection literature.
\section{Background and Preliminaries}

In this section, we introduce the notation and present some preliminaries of the intrusion detection setting, including the environment and the modeling of the defender and the intruder.

\noindent\textbf{Environment --} We model the protected environment as a graph $G = (N,E)$, where the node set $N$ represents distinct regions of the environment and the edge set $E$ represents feasible transitions available to the intruder between regions. An intruder path is modeled as a sequence of nodes in $G$ that respects the connectivity defined by $E$ (e.g., see Figure~\ref{fig:game}).

\noindent\textbf{Defender --} Let $\mathcal{V} = \{s_1, \hdots, s_p\}$, with $\abs{\mathcal{V}} = p$, denote the set of available sensors that can be placed on a subset of nodes in the graph. Let $\mathcal{L} \subseteq N$ denote the set of candidate placement nodes.  
We define the ground set of feasible sensor--location assignments as the Cartesian product
$\mathcal{Q} \triangleq \mathcal{V} \times \mathcal{L}$, where an element $e=(s,\ell)\in\mathcal{Q}$ represents placing sensor $s\in\mathcal{V}$ at location $\ell\in\mathcal{L}$. For simplicity, we refer to $\cals \subseteq \mathcal{Q}$ as a \textit{sensor placement (or sensor set, sensor deployment)} and elements of $\cals$ as \textit{sensors},  denoted by $s_q \in \cals$, which implicitly correspond to a sensor--location pair in $\mathcal{Q} = \mathcal{V} \times \mathcal{L}$.

\noindent\textit{Sensor's Action Space --} Each sensor has $d$ possible orientations (also referred to as actions), $\Theta = \{ \theta_1, \hdots, \theta_d\}$, with $\abs{\Theta} = d$. Each sensor covers a subset of nodes in the graph (i.e., regions of the environment)  depending on its orientation, which may not be restricted to its neighboring nodes.  Let $I = \{i_1, i_2, \hdots, i_m \}, m = |I|$, denote the set of joint orientation strategies of the defender. For a sensor set $\mathcal{S} \subseteq \mathcal{Q}$, the number of joint orientation strategies is $m = \abs{\Theta}^{\abs{\mathcal{S}}}$, which is exponential in the number of sensors in $\mathcal{S}$. 

\noindent\textit{Sensing Model -- }Each sensor $s_q \in \mathcal{V}$ has a probability of detection $p_{detect,q}$. Specifically, if the intruder traverses through a node which is covered by a sensor $s_q$, it will be detected by that sensor with a probability $p_{detect,q}$, and will go undetected by that sensor with probability $(1-p_{detect,q})$. We assume $p_{detect,q} \in [p_{min,q}, p_{max,q}]$, where $0<p_{min,q} < p_{max,q} < 1, q = 1,...,p$. Furthermore, a detection event from a sensor $s_q$ is independent of detection events from other sensors in $\mathcal{V}\setminus\{s_q\}$, conditioned on the intruder's path. This reflects realistic sensing conditions, where sensors are imperfect and detection probabilities are bounded by physical limitations, and sensor detections are typically modeled as conditionally independent events when operating under spatial separation and/or distinct sensing mechanisms~\cite{bhargav2025sensor}.

\noindent\textbf{Intruder --} We assume that the intruder can traverse the protected environment through multiple feasible paths, and that the set of such paths is finite. Let $J = \{j_1, j_2, \ldots, j_n\}$, with $n = |J|$, denote the set of all admissible paths. This set constitutes the intruder’s pure strategy space.

\begin{remark}
    If the set of intruder paths is not known a priori, it can be constructed from the underlying environment model. For example, one may enumerate admissible paths between designated origin and destination nodes using standard graph-theoretic path enumeration techniques. Alternatively, the path set may be obtained from reinforcement learning (RL) policies trained to navigate the environment, i.e., the trajectories or paths which are induced by the learned policy (e.g., see~\cite{ghasemi2021multiple}). 
\end{remark}

\section{ Sensor Orientation Scheduling}
In this section, we first aim to solve the scheduling problem for a given sensor placement by formulating the \orient~problem as a zero-sum game. We characterize key structural properties of the game’s payoff function and leverage them to design an efficient equilibrium-seeking algorithm with provable performance guarantees.

\subsection{Sensor Scheduling Problem Formulation}
For a given sensor placement, the goal of the defender is determine an optimal strategy to orient the sensors that minimizes the probability of not detecting the intruder. Let $\mathcal{S} \subseteq \mathcal{Q}$ denote the set of sensors placed (deployed in the environment). If a certain node is covered by multiple sensors in $\mathcal{S}$ and the intruder has traversed that node, then the overall probability of missed detection at that node is the product of the probability of missed detection of all the sensors covering that node. For a scheduling strategy $i\in I$ and a path $j \in J$, let $V^q_{ij}$ be the number of nodes that are covered by the sensor $s_q \in \mathcal{S}$ in the orientation strategy $i \in I$ and contained in the intruder's path $j \in J$. The overall probability of missed detection for the strategy pair $(i,j)$ is given by:
\begin{equation}
\label{eq:pmiss_prod}
    p_{miss}(i,j) = \prod_{q=1}^p (1-p_{detect,q})^{ V^q_{ij}}.
\end{equation}
We consider the negative log-likelihood representation for the probability of missed detection  (as multiplicative operations in \eqref{eq:pmiss_prod}  can lead to underflow and numerical instability), given by
\begin{equation}
\label{eq:log_p_miss}
    -\log(p_{miss}(i,j)) =  -\sum_{q=1}^p V^q_{ij}  \log (1-p_{detect,q}).
\end{equation}
We refer interested readers to \cite{krause2008near,breitung1991probability} for discussions on the use of log-likelihood representations for probability measures and their analytical advantages. We note that minimizing $p_{\text{miss}}(i,j)$ is equivalent to maximizing $-\log p_{\text{miss}}(i,j)$.

Let $A^{\cals},B^{\cals} \in \mathbb{R}^{m \times n}$ denote the payoff matrices for the defender and intruder, respectively, for the game induced by the deployed sensor set $\cals$. The payoff values for a pure strategy pair $(i,j)$ for the defender and intruder are given by $A^{\cals}_{ij} = - \log(p_{miss}(i,j))$ and $B^{\cals}_{ij} =   \log(p_{miss}(i,j))$, respectively. Since $A^{\cals} + B^{\cals} = 0$, the game is zero-sum, and without loss of generality it can be represented using only the payoff matrix $A^{\cals}$. Additionally, in a 2-player zero-sum game, each player may randomize (mix) over their available pure strategies; a mixed strategy is defined as a probability distribution over those strategies. Accordingly, we denote the defender’s and intruder’s mixed strategy spaces as 
$X = \Delta_m := \{ x \in \mathbb{R}^m_+ \mid \mathbf{1}^\top x = 1 \}$ 
and 
$Y = \Delta_n := \{ y \in \mathbb{R}^n_+ \mid \mathbf{1}^\top y = 1 \}$, 
respectively. For notational convenience, we also refer to the defender as the $x$-player and the intruder as the $y$-player. As in a standard zero-sum game with matrix $A^{\cals}$, the game's payoff value $v$, for a pair of mixed strategies $(x,y)$, is given by $v = x^{\top}A^{\cals}y$. 

\begin{definition}[Nash Equilibrium (NE)]
    A pair of mixed-strategies $(x^*,y^*)$ is a \textit{NE} of the zero-sum game with payoff matrix $A$ if $ x^{{* \top}}A y^* \geq x^{\top}Ay^*$ and $ x^{{* \top}}A y^* \leq x^{{* \top}}Ay, \hspace{2pt} \forall x \in X, \hspace{1pt} \forall y \in Y$. The resulting payoff $v^* = x^{* \top}Ay^*$ is called the optimal game value (or NE value).
    \end{definition}

In other words, a NE profile $(x^\star, y^\star)$ is a pair of mixed strategies such that neither player can improve their expected payoff by unilaterally deviating. The resulting value $v^\star = x^{\star \top} A y^\star$ is the best achievable utility for either of the players. We are now in place to formalize the \orient ~problem.

\begin{Problem}[\orient]
\label{prob:orient_prob}
Consider a given  sensor placement $\cals \subseteq \mathcal{Q}$ in the environment. The \orient~problem seeks to compute the NE strategies and value of the induced zero-sum game with payoff matrix $A^{\cals}$. Specifically, we aim to solve
\begin{equation}
\label{eq:zero-sum_game}
(x^\star, y^\star)
\in \argmax_{x \in X} \; \argmin_{y \in Y} \; x^\top A^{\cals} y,
\end{equation}
where $(x^\star, y^\star)$  denote the NE strategies and the optimal value of the game given by $
v(\cals ) = x^{\star \top} A^{\cals} y^\star$.
\end{Problem}

The solution to Problem~\ref{prob:orient_prob} admits a natural operational interpretation. At regular scheduling cycles (e.g., daily or hourly reconfiguration cycles), the defender samples an orientation configuration according to the equilibrium distribution $x^\star$ and deploys it. Over a sufficiently long horizon, no alternative scheduling policy can guarantee a strictly higher detection performance  against a rational intruder. In this sense, the equilibrium strategy provides a maximin-optimal and robustness certificate for persistent deployment. Additionally, the scheduling strategy $x^\star$, computed via Algorithm~\ref{alg:orient} hedges against the most rational intruder (best-response). If the intruder's behavior is adversarial or deviates from being fully rational (e.g., bounded rationality), the realized detection performance can only improve for the defender~\cite{freund1996game}. This follows directly from the concept of a NE -- no player can improve their payoff by unilaterally deviating from their equilibrium strategy.

\subsection{Efficient Algorithm for Solving the Game}

We propose a fast and scalable algorithm that leverages the structure of the game matrix $A^\cals$ for efficiently solving for the NE. For a moderate number of strategies, the optimization problem in \eqref{eq:zero-sum_game} can be solved to compute the NE strategies using a bi-level Linear Programming (LP) formulation \cite{nehme2009two}. However, in the \orient~problem, the defender has an exponentially sized set of scheduling strategies, i.e.,  $m = |\Theta|^{|\cals|}$. For example, if $d=\lvert\Theta \rvert = 4$ (i.e., each sensor has 4 possible orientations) and $\lvert\mathcal{S}\rvert=10$ (i.e., there are $10$ sensors), then $\lvert I\rvert = 4^{10}=1,048,576$.  Solving for the exact NE using an LP can become computationally intractable due to a large number of variables and constraints. To this end, we present the Distributed Equilibrium Seeking (DES) algorithm (Algorithm \ref{alg:orient}). This algorithm is a variant of the \textit{Weighted Majority algorithm (WMA)}  originally proposed in \cite{freund1996game}. The DES algorithm leverages the structure of the game to perform multiplicative weight updates locally for each sensor, which will then be aggregated to estimate the joint scheduling strategy for the defender.

In WMA \cite{freund1996game}, one of the players (say, the $y$-player) maintains a weight vector $w$ over its pure strategies and updates it iteratively based on feedback from the interaction with the $x$-player. At iteration $t$, the $x$-player computes a best response $x_t = \argmin_{x \in X} x^\top A y_t$ to the current mixed strategy $y_t$ (by solving an LP), where $A$ is the game matrix. The induced loss vector for the $y$-player is given by $l_t = x_t^\top A$. The $y$-player then updates its weights using a multiplicative rule of the form $w_{t+1}(j) = w_{t}(j) \beta^{- l_{t}(j)}$, where $\beta \in (0,1)$ is a  multiplicative penalty parameter, and updates the mixed strategy as $y_{t+1} = w_{t+1}/{\mathbf{1}^\top w_{t+1}}$. Asymptotically, i.e., as the total number of iterations $T \to \infty$, the average strategies $\bar{x}_T = \frac{1}{T} \sum_{t=1}^T x_t$ and $\bar{y}_T = \frac{1}{T} \sum_{t=1}^T y_t$ converge to a NE of the zero-sum game. In particular, the average strategies $(\bar{x}_T, \bar{y}_T)$ approach the set of equilibrium strategies, and the average payoff $\frac{1}{T} \sum_{t=1}^T x_t^\top A y_t$ converges to the game value.

\begin{remark}
One may naturally ask why the standard WMA cannot be applied directly to the \orient~problem. As indicated earlier, the challenge arises from the size of the defender’s joint strategy space, which grows as $\mathcal{O}(|\Theta|^{|\cals|})$ with the number of deployed sensors $|\cals|$. Consequently, computing an exact best response at each iteration would be computationally prohibitive.
\end{remark}

Now, we define the set of sub-game payoff matrices $\mathcal{A} = \{ A^1, A^2, \hdots, A^{|\cals|} \}$, where $A^q \in \mathbb{R}^{d \times n}$ corresponds to sensor $s_q \in \mathcal{S}$, given by
\begin{equation}
    \label{eq:small_A}
    A^q_{kj} =
     - V^q_{kj} \times \log(1-p_{detect,q}),
\end{equation}
where $V^q_{kj}$ contains all the nodes that are covered by the sensor $s_q$ oriented in the direction $\theta_k$ and contained in the intruder path $j$. The matrix $A^q$ contains the payoffs/log-probability of detection associated with sensor $s_q$'s interaction with the intruder in isolation. This can be effectively seen as the payoff matrix of a game between the intruder and the sensor $s_q$ alone.
In Algorithm~\ref{alg:orient}, the defender estimates the best response against the intruder’s mixed strategy in a distributed manner, using these sub-game payoff matrices,  rather than solving a single large LP at every iteration.

\begin{algorithm}[!t]
\caption{ Distributed Equilibirum Seeking (DES)}\label{alg:orient}
\begin{algorithmic}[1]
\Require{ Game matrix $A^\cals$, Set of sub-game matrices $\mathcal{A} = \{A^1,\hdots, A^{|\cals|}\}$, $\beta \in (0,1)$, Number of Iterations $T$}
\Ensure  {Mixed-strategies $(\bar{x}, \bar{y})$ and value $\bar{v}$}
\State \textbf{Initialize:} $w_1(j) = 1/|J|, j = 1,...,n$

\For {$t = 1, \hdots, T$}
\State Obtain $y_t = w_t/\mathbf{1}^\top w_t$
\State Solve $\sigma^{ q \star}_t \in \argmax_{\sigma \in \Delta_{|\Theta|}} \sigma^{\top}A^{q}y; \hspace{2pt} q = 1, \hdots, |\cals|$ 
\State Estimate joint-strategy for sensors
\begin{align}
    x_{t}(i) = \prod_{q = 1}^{p} \sigma^q_{t}(k_{q,i}), \forall i \in I
\end{align}
\State Compute intruder's loss $\ell_t(\cdot) = x_t^{\top} A^{\cals}$ (as in \eqref{eq:intruder_loss})
\State Update intruder's weights
\begin{equation}
    \label{eq:dist_weight}
    w_{t+1}(j) = w_{t}(j) \beta^{-\ell_t(j) }; \quad j = 1,\hdots,|J|
\end{equation}
\EndFor

\State \Return $(\bar{x},\bar{y})  = (\frac{1}{T} \sum_{t=1}^T x_t, \frac{1}{T} \sum_{t=1}^T y_t);\hspace{2pt} \bar{v} = \bar{x}^\top A^\cals \bar{y}$

\end{algorithmic}
\end{algorithm}

Let $\sigma^q \in \Delta_{|\Theta|}$ denote the mixed strategy of sensor $s_q$ over its set of orientations $\Theta$, where $\Delta_{|\Theta|} := \{ \sigma \in \mathbb{R}^{|\Theta|}_+ \mid \mathbf{1}^\top \sigma = 1 \}$ denotes the $|\Theta|$-dimensional probability simplex. In line 4 of Algorithm~\ref{alg:orient}, the defender solves an LP to obtain the best-response mixed-strategy $\sigma^{\star q}_t \in \argmax_{\sigma \in \Delta_{|\Theta|}} \sigma^{\top}A^{q}y_t$, for each sensor $s_q \in \cals$\footnote{We note that Line 4 of Algorithm~\ref{alg:orient} admits a fully-distributed implementation. 
The defender can compute the per-sensor best responses independently and in parallel (e.g., via multi-threaded or distributed processing), thereby improving scalability with respect to $|\cals|$.}. Then, the overall mixed-strategy $x_t$ is obtained by the product distribution 
$ x_{t}(i) = \prod_{q = 1}^{p} \sigma^q_{t}(k_{q,i}), \forall i \in I$, where $k_{q,i}$ refers to the orientation of sensor $s_q$ in the joint strategy $i$.  For notational convenience, we denote the product distribution of the strategies $\sigma^q$ by 
$x = \bigotimes_{s_q \in \cals} \sigma^q$. The intruder's loss vector is then given by
\begin{equation}
    \label{eq:intruder_loss}
    \ell_t = x_t^\top A^{\cals},
\end{equation}

\noindent which is used to update the intruder's weights in Line 7 of Algorithm~\ref{alg:orient}. We now have the following result showing that the defender’s joint best response decomposes across sensors. 

\begin{theorem}
\label{thm:dist_best_response}
Let $\sigma^{q\star} \in \argmin_{\sigma^q \in \Delta_{|\Theta|}} \sigma^{q\top} A^q y$, 
 denote the best response (mixed-strategy)  of sensor $s_q \in \cals$ with respect to its sub-game matrix $A^q$ against a given intruder's mixed-strategy $y \in Y$. Define the joint mixed-strategy of the defender as  
$x^\star = \bigotimes_{s_q \in \cals} \sigma^{q\star}$. 
Then $x^\star$ satisfies
\[
x^\star \in \argmin_{x \in X} x^\top A^{\cals} y.
\]
\end{theorem}
In other words, the product of the individual best responses $\sigma^{q \star}$ computed per sensor, using their corresponding sub-game matrices $A^q, \hspace{2pt} q=1, \hdots, |\cals|$, is a best response mixed strategy for $y$ with respect to the full game matrix $A^{\cals}$. We defer all detailed theoretical proofs to the Appendix. 

\begin{remark}
 A centralized best response requires solving an LP over $|\Theta|^{|\cals|}$ variables, with the  complexity under interior-point methods scaling as $\mathcal{O}(|\Theta|^{3|\cals|})$ \cite{potra2000interior}. In contrast,  Algorithm \ref{alg:orient} requires solving $|\cals|$ independent LPs, each with $|\Theta|$ variables, resulting in total complexity in the order of $\mathcal{O}(|\cals|\,|\Theta|^{3})$. This is an exponential reduction in complexity with respect to $|\cals|$.
\end{remark}

Now, combining Theorem~\ref{thm:dist_best_response} with the results in Sections 2.4 and 2.5 of \cite{freund1996game}, we have the following result.

\begin{theorem}
\label{thm:nashgap}
For a given sensor placement $\cals \subseteq \mathcal{Q}$, let $(\bar{x}, \bar{y})$ and $\bar{v}$ denote the mixed strategies and the corresponding game value returned by Algorithm~\ref{alg:orient} after $T$ iterations, with 
\begin{equation}
\label{eq:beta_wma}
\beta = \left(1 + \sqrt{\frac{2 \ln |J|}{\tilde{L}}} \right)^{-1},
\end{equation}
where $\tilde{L} \ge \min_{j \in J} \sum_{t=1}^{T} \ell_t(j)$. Let $v(\cals)$ denote the NE value of the game with payoff matrix $A^{\cals}$ (i.e., the optimal value of Problem~\ref{prob:orient_prob}). Then,
\begin{equation}
\label{eq:nashgap}
 v(\cals) - \bar{v} 
\le
\underbrace{
\frac{\sqrt{2 \tilde{L}\,\ln |J|}}{T}
+
\frac{\ln |J|}{T}
}_{\epsilon_T}.
\end{equation}
\end{theorem}

The quantity $\Tilde{L}$ in Theorem \ref{thm:nashgap} is an upper bound on the sum of losses over $T$ iterations, which can be loosely bounded by $T$. Combining this bound with \eqref{eq:nashgap}, we have that $\epsilon_T = \mathcal{O}(1/\sqrt{T})$. Therefore, we have $\epsilon_T \to 0$ as $T \to \infty$. Consequently, $\bar{v} \to v(\cals)$, and $(\bar{x}, \bar{y})$ converge to a pair of NE strategies of the game. Additionally, the convergence rate of $\epsilon_T$ in~\eqref{eq:nashgap} depends on $\ln |J|$ and is independent of the number of sensors. In contrast, in our prior work~\cite{bhargav2025sensor}, the convergence rate was dependant on $|S|\ln|\Theta|$, growing linearly with the number of sensors; consequently, as the system size increases, this dependence leads to slower convergence.

\begin{remark}
The choice of $\beta$ in~\eqref{eq:beta_wma} depends on $\tilde{L}$, which may not be known or easily bounded in practice. In such cases, we adopt the following estimation procedure. First, we set $\tilde{L} = T$. For a desired $\epsilon_T$, we estimate $T$ by solving~\eqref{eq:nashgap}. We then use this estimate of $T$ as $\tilde{L}$ in~\eqref{eq:beta_wma} to obtain $\beta$.
\end{remark}
\subsection{Runtime Analysis}
We now discuss the computational advantage of Algorithm~\ref{alg:orient}. Consider the scenario in which the game is solved via a centralized (bi-level) LP formulation. We assume that the LP is solved using standard interior-point methods \cite{potra2000interior}. It is important to note that interior-point methods operate in finite-precision arithmetic and therefore return $\epsilon$-optimal solutions (i.e., with duality gap bounded by $\epsilon$), rather than exact symbolic optima. In practice, exact solutions cannot be obtained due to numerical tolerances and termination criteria. Under standard complexity analysis \cite{potra2000interior}, interior-point methods require $\mathcal{O}(\log(1/\epsilon))$ iterations to reach an $\epsilon$-accurate solution, with each iteration incurring computational cost  of $\mathcal{O}((|\Theta|^{|\cals|} + |J|)^3)$ for the centralized formulation. Consequently, the overall runtime required to obtain an $\epsilon$-optimal solution via the centralized LP formulation scales as
\[
T_{LP}(\epsilon) \triangleq \mathcal{O}\!\left((|\Theta|^{|\cals|} + |J|)^3 \log(1/\epsilon)\right).
\]
In contrast, Algorithm~\ref{alg:orient} requires
\[
T_{\text{alg}}(\epsilon) \triangleq \mathcal{O}\!\left(|\cals|\,|\Theta|^3/\epsilon^2\right)
\]
operations to compute an $\epsilon$-accurate approximation of the game value.  As $|\cals|$ increases, the exponential term $|\Theta|^{|\cals|}$ dominates the polynomial dependence in $T_{\text{alg}}(\epsilon)$, implying that $T_{LP}(\epsilon)$ grows exponentially in the number of sensors, whereas $T_{\text{alg}}(\epsilon)$ scales only linearly in $|\cals|$ and polynomial in $|\Theta|$. Hence, for moderate to large $|\cals|$, the centralized LP-based approach becomes computationally prohibitive, while Algorithm~\ref{alg:orient} maintains scalable performance.

\begin{remark}
    Our framework admits a significant implementation advantage. If a central coordinator computes and broadcasts the intruder mixed strategy $y_t$ at each iteration, no sensor-to-sensor communication is required, and no local weight vectors or orientation distributions need to be exchanged. Given $y_t$, each sensor updates its local orientation distribution independently using only its own sub-game matrix. Importantly, $y_t$ is a single broadcast vector whose dimension depends solely on the intruder strategy space. Consequently, the per-iteration communication cost scales linearly with the size of the intruder strategy set and remains independent of the exponentially large joint orientation space. This separation preserves global equilibrium guarantees while ensuring communication-efficient and scalable execution.
\end{remark}

\section{Sensor Placement Optimization}

In this section, we turn our attention to the meta-problem of optimal sensor placement -- the \sense~problem. We first formalize the placement problem by connecting it to the \orient~problem discussed in the previous section and then characterize the structural properties of its objective function. Leveraging these properties, we develop an efficient sensor placement algorithm with provable performance guarantees.

\subsection{Sensor Placement Problem Formulation}

Recall that the ground set of feasible sensor--location assignments is the Cartesian product
$\mathcal{Q} = \mathcal{V} \times \mathcal{L}$, where an element $e=(s,\ell)\in\mathcal{Q}$ represents placing sensor $s\in\mathcal{V}$ at location $\ell\in\mathcal{L}$.  We assume uniform sensor selection costs and  denote the total deployment budget by $B \in \mathbb{Z}_+$. In addition, to ensure physical feasibility and realizability, we enforce that each sensor can be assigned to at most one location. These constraints can be captured via a matroid intersection over $\mathcal{Q}$.
Specifically, we define the partition matroid $\mathcal{M}=(\mathcal{Q},\mathcal{I}_{loc})$, where
$\mathcal{I}_{loc} := \{ \cals \subseteq \mathcal{Q} \mid |\cals \cap (\{s\}\times\mathcal{L})| \le 1,\ \forall s\in\mathcal{V} \}$, which enforces that each sensor is assigned to at most one location.

The set of feasible deployments is then given by
\begin{equation}
\label{eq:matroid_constr}
    \mathcal{I} \;:=\; \{ \cals \subseteq \mathcal{Q} \mid \cals \in \mathcal{I}_{loc},\ |\cals| \le B \}.
\end{equation}
Equivalently, $\mathcal{I}$ is the intersection of the partition matroid $\mathcal{I}_{loc}$ and a uniform matroid corresponding to the budget constraint $B$. We now formalize the sensor placement problem below.

\begin{Problem}[\sense]
\label{prob:sense_prob}
     Let $v: 2^\mathcal{Q} \to \mathbb{R}_{\ge 0}$ denote the optimal game value obtained by solving the \orient~problem for a placement $\cals \subseteq \mathcal{Q}$.  The \sense~ meta-problem aims to identify  $\cals^\star \subseteq \mathcal{Q}$ that solves the following optimization problem.
\[
\cals^\star = \argmax_{\cals \in \mathcal{I}} \; v(\cals).
\]
\end{Problem}

Problem \ref{prob:sense_prob} is a combinatorial subset selection problem. Such problems are generally NP-hard~\cite{bhargav2023complexity, ye2020complexity}, and computing the exact (optimal) solution can be computationally expensive. A practical alternative is therefore to seek an approximate solution~\cite{bhargav2023complexity, xu2024performance, hashemi2020randomized}. To this end, we propose a greedy algorithm that efficiently approximates the solution of Problem~\ref{prob:sense_prob}.

\subsection{Efficient Algorithm for Sensor Placement}

We first begin by characterizing the structural properties of the utility function of Problem~\ref{prob:sense_prob}.
For a set function $f$, let $f(i \mid \mathcal{T}) \triangleq f(\mathcal{T} \cup \{i\}) - f(\mathcal{T})$ denote the marginal gain of adding element $i$ to set $\mathcal{T}$. 
\begin{definition}[Monotonicity \cite{hashemi2019submodular}]
A set function $f : 2^{\mathcal{X}} \rightarrow \mathbb{R}$ is  monotone non-decreasing if $f(\mathcal{S}) \leq f(\mathcal{T}), \hspace{2pt} \forall \mathcal{S} \subseteq \mathcal{T} \subseteq \mathcal{X}.$
\end{definition}

\begin{definition}[Additive Weak-Submodularity \cite{ hashemi2019submodular}]
A set function 
$f : 2^{\mathcal{X}} \rightarrow \mathbb{R}$ satisfies additive weak-submodularity if 
\[
f(i \mid \mathcal{T}) - f(i \mid \mathcal{S})
\;\le\;
\epsilon_f, \quad
\forall \mathcal{S} \subseteq \mathcal{T} \subseteq \mathcal{X},\;
i \in \mathcal{X} \setminus \mathcal{T},
\]

\noindent where $\epsilon_f$ is the weak-submodularity constant defined as

\begin{equation}
\label{eq:submod_const}
\epsilon_f 
=
\max_{\substack{\mathcal{S} \subseteq \mathcal{T} \subseteq \mathcal{X} \\ i \in \mathcal{X} \setminus \mathcal{T}}}
\Big( f(i \mid \mathcal{T}) - f(i \mid \mathcal{S}) \Big).
\end{equation}

\end{definition}

\begin{remark}
    When leveraging weak-submodularity in subset selection problems where one typically operates under a cardinality budget, e.g., selecting a $k$-sized subset, it suffices to characterize the weak-submodularity constant in \eqref{eq:submod_const} over sets of size less than $k$  (i.e., $|\mathcal{T}| < k$), since these are the only sets encountered during the optimization process~\cite{khanna2017scalable}.
\end{remark}

For each sensor $s_q \in \mathcal{Q}, \hspace{2pt} q = 1, \hdots, |\mathcal{Q}|$, we define 
\begin{equation}
\label{eq:lambda_bound}
   \displaystyle\xi({s_q}) \triangleq \max_{\theta_k \in \Theta} \max_{j \in J} \left(-V^q_{kj} \times \log(1-p_{detect,q})\right). 
\end{equation}

In other words, $\xi(s_q)$ quantifies the maximum absolute contribution that sensor $s_q$ can make to the game’s utility, taken over all feasible orientations $\theta_k \in \Theta$ and all possible intruder paths $j \in J$. Thus, $\xi(s_q)$ provides an upper bound on the contribution of $s_q$ to any entry of the game payoff matrix.

\begin{lemma}
\label{lma:weaksubmod}
Let $v:2^{\mathcal{Q}}\to\mathbb{R}_{\ge 0}$ denote the objective function of the \sense~problem and $B$ denote the placement budget. The function $v$ has the following properties. 

\textit{(i) Monotonicity --} For all $\cals \subseteq \cal{T}\subseteq\mathcal{Q}$,
\[
v(\cals)\;\le\; v(\cal{T}).
\]

\textit{(ii) Weak-submodularity --}
For all $\cals \subseteq \cal{T}\subseteq\mathcal{Q}$ with $|\mathcal{T}|\le B$ and all $i\in\mathcal{Q}\setminus \mathcal{T}$,  $v$ is additive weak-submodular with constant 
\begin{align}
\label{eq:submod_const_bound}
\epsilon_{v,B}\;&\le\; 2 
\sum_{s \in \mathcal{S}^B}\; \xi(s),
\end{align}
where $\cals^B$ denotes the subset of $B$ sensors with the largest sum of the 
$\xi(s)$ values.
\end{lemma}

We now present Algorithm~\ref{alg:sense}, a greedy procedure for approximating the solution to the \sense~problem. The algorithm initializes with an empty  set, $\cals^{g} = \emptyset$, and iteratively adds the sensor that yields the largest marginal improvement in the utility. At each iteration (Line~3), the value $v(\cdot)$ is evaluated by solving the corresponding \orient~problem using Algorithm~\ref{alg:orient}.\footnote{We assume that Algorithm~\ref{alg:orient} returns the optimal solution to the \orient~problem; in practice, this is done by choosing the approximation parameter $\epsilon_T$ sufficiently close to zero (e.g., $\epsilon_T = 1e-6$).} The procedure terminates once the placement budget is reached, i.e., when $|\cals^{g}| = B$.

\begin{algorithm}[h!]
    \caption{Greedy Algorithm for \sense \label{alg:sense}}
    \begin{algorithmic}[1]
      \Require   Candidate sensors $\mathcal{Q}$, Placement budget $B$
       \Ensure Sensor placement $\cals^g \subseteq \mathcal{Q}$ such that $|\cals^g| \leq B$
            \State $\cals^g \gets \emptyset$
            \While {$ |\cals^g| \le B$}
                \State Identify the sensor with largest utility gain: \begin{equation}
                \label{eq:greedy_step}
                     i^* \gets \argmax_{i \in \mathcal{Q} \setminus \cals^g} \hspace{2pt} \left(v(\cals^g \cup \{i\}) - v(\cals^g)\right)
                \end{equation}
                \State $\cals^g \gets \cals^g \cup \{i^*\}$
            \EndWhile\\
    \Return $\cals^g$ 
    \end{algorithmic}
\end{algorithm}

We have the following guarantees for Algorithm~\ref{alg:sense}, which follows from the near-optimality bounds presented in~\cite{fisher2009analysis,calinescu2011maximizing,hashemi2019submodular}

\begin{theorem}
\label{thm:greedy}
    For a given set of candidate sensors $\mathcal{Q}$ and placement budget $B \in \mathbb{Z}_+$, let $\sopt$ denote the optimal solution (sensor placement) to the \sense~problem (Problem \ref{prob:sense_prob}), and let $\cals^g$ denote the sensor placement returned by Algorithm~\ref{alg:sense}. We have
    \begin{equation}
        \label{eq:greedy_bound}
        v(\cals^g) \ge \ \frac{1}{2} \hspace{2pt} \left( v(\sopt) - (B-1) \epsilon_{v,B} \right),
    \end{equation}
    where $\epsilon_{v,B}$ is the weak-submodularity constant as defined in~\eqref{eq:submod_const_bound}.
\end{theorem}

\begin{remark}
   The multiplicative factor of $1/2$ in the near-optimality bound of Theorem~\ref{thm:greedy} can be improved to $(1 - 1/e)$ by employing a continuous version of the greedy algorithm over the multi-linear extension of the objective combined with pipage (or swap) rounding to obtain the discrete solution, originally developed by ~\cite{calinescu2011maximizing}. We resort to the vanilla-greedy scheme (Algorithm~\ref{alg:sense}) for convenience and computational simplicity. Moreover, we empirically demonstrate that even the simple greedy scheme (Algorithm~\ref{alg:sense})  achieves utilities that are very close to the optimal value. Additionally, the $(B-1) \epsilon_{v,B}$ term is the suboptimality arising from the weak-submodularity of the objective. We note that the derived bound on $ \epsilon_{v,B}$ in \eqref{eq:submod_const_bound} is inherently conservative, as it is obtained under worst-case assumptions. In practice, the effective value of $ \epsilon_{v,B}$ can be typically much smaller (see Figure~\ref{fig:subopt_bound_compare_sense}) and depends on the specific problem instance, particularly the structural properties of the candidate sensor placement locations and intruder paths.
\end{remark}

We now state the following corollary, which directly follows from Theorem~\ref{thm:greedy} in the special case where all sensors are identical, i.e., they have the same $p_{detect}$, orientation action space $\Theta$,  and the same coverage/range (i.e., same $V^q_{ij}$'s).
\begin{corollary}
\label{coro:same_sensors}
   If all sensors are identical,  Algorithm~\ref{alg:sense} has the following guarantees.
\[
v(\cals^{g}) \;\ge\; \left(1 - \frac{1}{e}\right) \left( v(\sopt) - (B-1)\,\epsilon_{v,B}\right).
\] 
\end{corollary}
\section{Experiments}

In this section, we present an extensive set of simulation studies designed to assess the performance of the proposed framework in an intrusion detection scenario.

\subsection{Experimental Setup}
We begin by describing the experimental setup in detail, outlining the underlying environment, the modeling assumptions, and strategy spaces of both the defender and the intruder.

\noindent\textbf{Environment --}
We consider the nine-room grid world environment~\cite{ghasemi2021multiple}, illustrated in Figure~\ref{fig:env_sensors}. The environment consists of a $19 \times 19$ grid of cells, which can be equivalently represented as a graph $G(V,E)$ with $|V| = 361$ nodes, each corresponding to a cell in the grid. The edges in $E$ capture adjacency relationships between the cells/nodes. Cells marked in black denote obstacles and are impassable. Sensors can only be installed at the candidate locations highlighted in red (and indexed in Figure~\ref{fig:env_sensors}). We denote the set of candidate placement nodes (or cells) by $\mathcal{L}$.

\begin{figure}[!h]
    \centering
\includegraphics[width=0.35\linewidth]{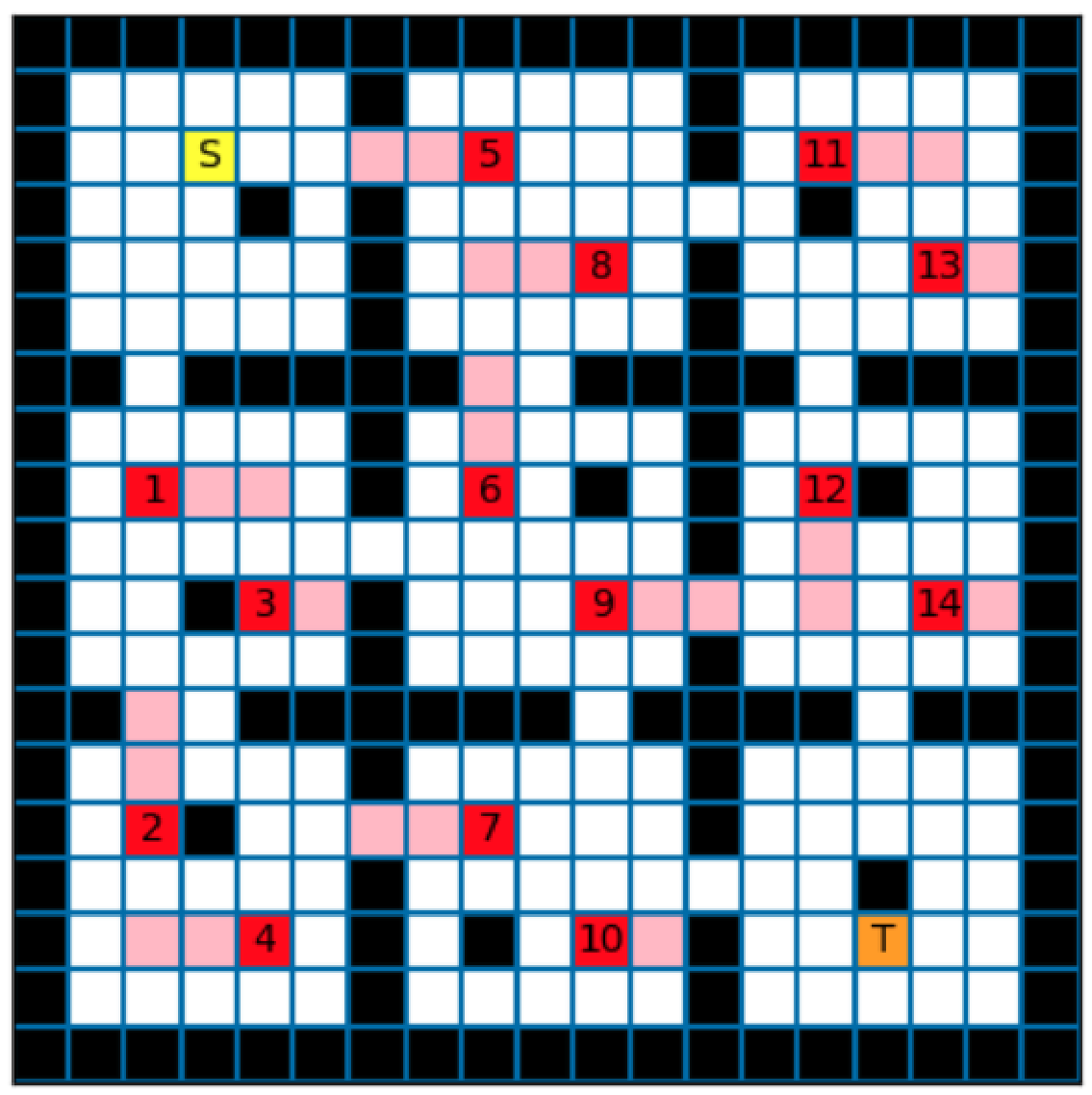}
    \caption{The nine-room grid world environment --  Candidate sensor placement locations are indexed from $1$ to $14$. A representative pure orientation-scheduling strategy of the defender, shown for all $14$ sensors, is illustrated. Each sensor's field-of-view consists of $3$ adjacent cells (including the cell in which it is placed) along the direction of its orientation.}
    \label{fig:env_sensors}
\end{figure}

\noindent\textbf{Defender -- }
We denote the ground set of sensors available to the defender as $\mathcal{V}$. The set of sensor-location pairs available to the defender (to determine the optimal placement) is thus $\mathcal{Q} = \mathcal{V} \times \mathcal{L}$.  Each sensor in $\mathcal{V}$ has a discrete orientation set with cardinality $\abs{\Theta}=4$, corresponding to the four cardinal directions -- $\{\scenario{East}, \scenario{North}, \scenario{West}, \scenario{South}\}$. For a selected sensor subset $\mathcal{V'} \subseteq \mathcal{V}$, the defender’s (pure) orientation strategy space consists of all joint orientation assignments and therefore has size $4^{\abs{\mathcal{V'}}}$.  For each sensor $s_q \in \mathcal{V}$, we independently sample its detection probability $p_{\text{detect}}^{q}$ uniformly at random from the interval $(0.2, 0.8)$ (unless specified otherwise)  and fix (freeze) this value for the entire problem instance. This models heterogeneity in sensing quality across candidate sensors while ensuring that each instance remains internally consistent.
We treat both the size of the ground set $\abs{\mathcal{V}}$ and the sampled detection probabilities $\{{p_{\text{detect}}^{q}} \mid {s_q \in \mathcal{V}}\}$ as instance-level parameters. Varying $\abs{\mathcal{V}}$ (which controls combinatorial complexity) and resampling detection probabilities (which affects payoff structure), allows us to generate multiple problem instances with differing strategic complexity and heterogeneous utilities.   

\noindent\textbf{Intruder --}
We assume the intruder originates at the source node $S$ (marked in yellow) and aims to reach the target node $T$ (marked in orange) (see Figure~\ref{fig:intruder_paths}). To model strategic diversity, we explicitly enumerate a set $J$ of $\abs{J} = 50$ feasible source-to-target paths. These paths are constructed to be structurally diverse, capturing multiple ways of reaching $T$ by traversing different rooms and corridors of the nine-room environment, while respecting obstacle constraints (see~\cite{ghasemi2021multiple} for more details on techniques for enumerating such diverse paths). 

\begin{figure}[!h]
    \centering
    \includegraphics[width=0.6\linewidth]{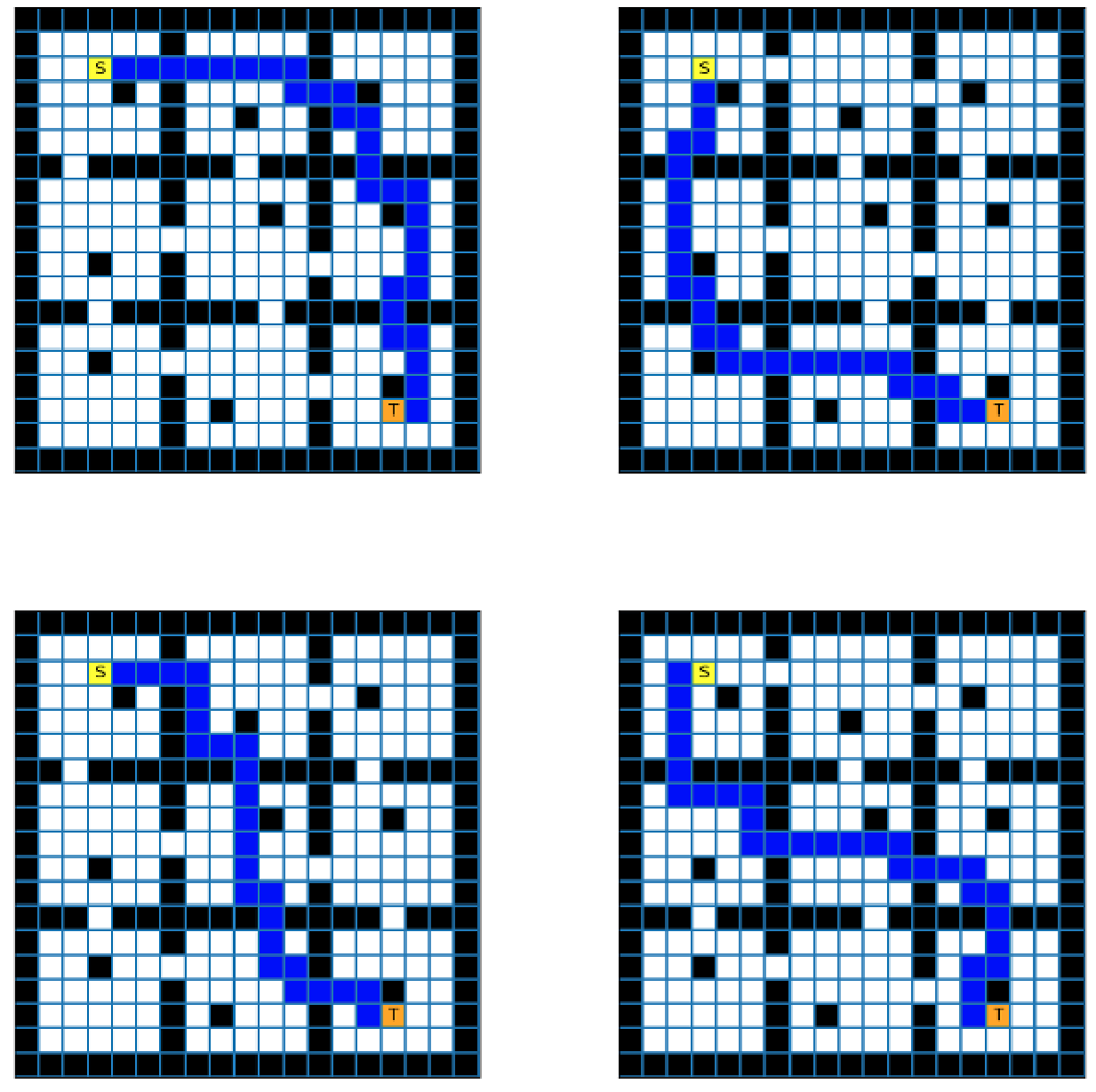}
    \caption{Sample intruder paths (pure strategies) in the nine-room grid world environment. All paths originate at node $S$ (marked in yellow) and terminate at $T$ (marked in orange).}
    \label{fig:intruder_paths}
\end{figure}

% \noindent All simulations are run on a 2.6 GHz 6-Core Intel i7 machine. 

\subsection{Convergence and Scalability of Algorithm~\ref{alg:orient}}

We evaluate the convergence behavior and scalability of Algorithm~\ref{alg:orient} for solving the \orient~problem to a prescribed approximation accuracy. 

For the convergence study, we generate instances of the \orient~problem with $|\mathcal{V}| = 5$. For each sensor $s_q \in \mathcal{V}$, the detection probability $p_{\text{detect}}^{q}$ is sampled independently and uniformly at random from the interval $(0.5, 0.8)$. to generate 20 game instances with slightly perturbed payoff matrices\footnote{We normalize the payoff matrix entries to be in $(0,1)$. Note that normalization preserves NE strategies and scales the game value to within $(0,1)$~\cite{bhargav2025sensor}.}. Algorithm~\ref{alg:orient} is executed on each instance. Figure~\ref{fig:game_nash} reports the average defender and intruder payoff trajectories across the 20 instances, with shaded regions indicating one standard deviation around the mean. We observe that the payoffs converge rapidly toward the equilibrium value, reaching negligible deviation within a few hundred iterations.

For scalability studies, we assess the runtime performance of Algorithm~\ref{alg:orient} as the problem size increases, while maintaining a prescribed equilibrium approximation quality against the following baselines:
\begin{itemize}
\item Solving the standard linear programming (LP) formulation of the game~\cite{nehme2009two};
\item The Weighted Majority Algorithm (WMA) proposed in~\cite{freund1996game}.
\end{itemize}
To generate instances of increasing complexity, we vary the number of sensors $\abs{\mathcal{V}} \in \{2,3,4,5,6,7,8\}$, which induces exponential growth in the defender’s strategy space, while keeping the intruder’s strategy space fixed at $|J| = 50$. We set the target equilibrium accuracy to $\epsilon_T = 0.001$, i.e., we compute an $\epsilon$-NE with $\epsilon = 0.001$. For each instance, the required number of iterations $T$ to achieve an $\epsilon_T$-NE is estimated using~\eqref{eq:nashgap}, with $\Tilde{L} = T$.
Each game instance is solved using: (i) the standard LP formulation~\cite{nehme2009two} implemented via the Gurobi Optimization Solver, (ii) WMA~\cite{freund1996game}, and (iii) Algorithm~\ref{alg:orient}. Due to solver limitations under the Gurobi academic license, we were unable to solve game instances with strategy spaces exceeding $2000$ using the LP formulation. The compare the corresponding running times in Figure~\ref{fig:lp_wma_des_runtime}. We observe that both WMA and Algorithm~\ref{alg:orient} attain the prescribed $0.001$-approximate equilibrium substantially faster than the LP formulation. Furthermore, as the defender’s strategy space grows, Algorithm~\ref{alg:orient} demonstrates superior scalability compared to WMA, with increasingly pronounced runtime gains.

% \begin{table}[!h]
%     \centering
%     \begin{tabular}{ccccc}
%     \hline
% $|\mathcal{S}|$ & Strategy Spaces & LP & WMA & \textbf{Algorithm~\ref{alg:orient}} \\
%        \hline 
%        \hline
%        $2$ & $16 \times 50$ & 3.22  & 2.73  & 1.72 \\
%          $3$ &  $64 \times 50$ & 7.90  & 3.04  & 2.35 \\
%         $4$ &  $256 \times 50$ & 24.88  & 6.09  & 3.77 \\
%         $5$ &  $1024 \times 50$ & 75.72  & 12.15  & 5.34 \\
%        $6$ &   $4096 \times 50$ & - & 35.51  & 6.91 \\
%       $7$ &    $16,384 \times 50$ & - & 107.65  & 8.21 \\
%         $8$ &  $65,536 \times 50 $ & - & 365.43  & 9.35 \\
%          \hline
%     \end{tabular}
%     \caption{Comparison of running times (in secs) of Distributed Equillibirum Seeking Algorithm (Algorithm \ref{alg:orient}) with standard Weighted Majority algorithm (WMA)~\cite{freund1996game} and Linear Program}
%     \label{tab:table1}
% \end{table}
\begin{figure*}[!ht]
    \centering
    \begin{subfigure}{0.3\textwidth}
    \centering        \includegraphics[width=155pt]{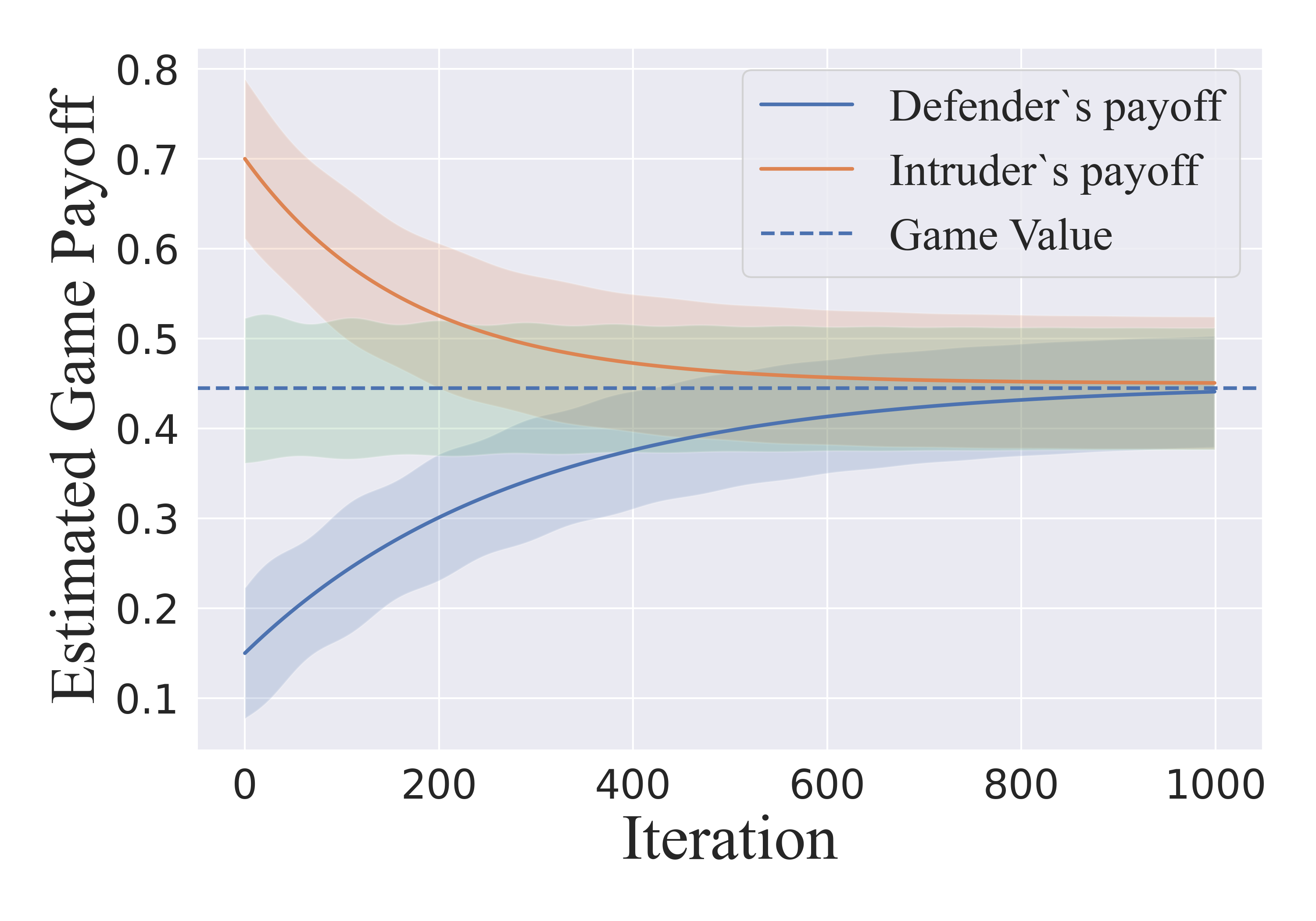}
        \caption{Convergence of defender's and intruder's value estimates to the Nash equilibrium value of the scheduling game}
        \label{fig:game_nash}
    \end{subfigure}
    \hfill
        \begin{subfigure}{0.3\textwidth}
    \centering        \includegraphics[width=160pt]{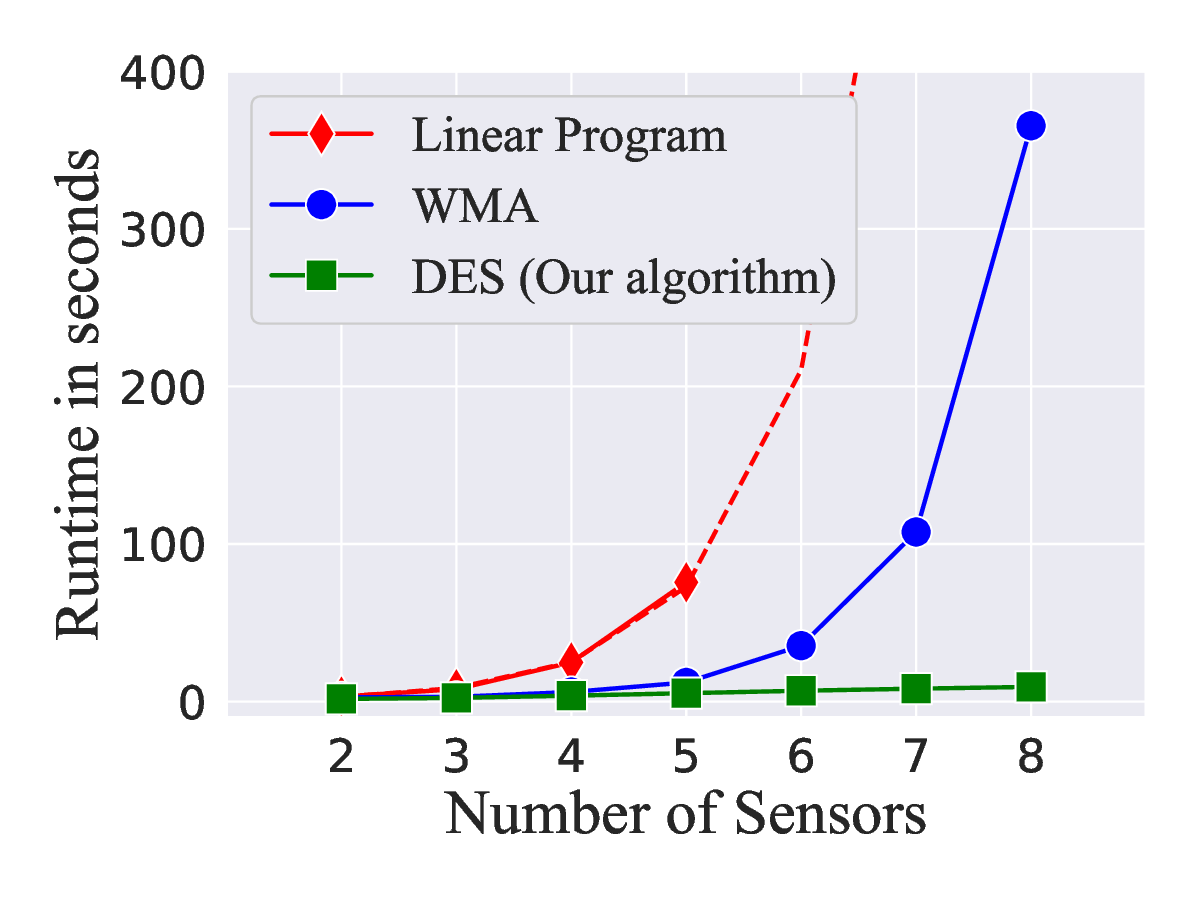}
        \caption{Runtime comparision of DES with LP and WMA}
        \label{fig:lp_wma_des_runtime}
    \end{subfigure}
    \hfill
    \begin{subfigure}{0.3\textwidth}
    \centering        \includegraphics[width=145pt]{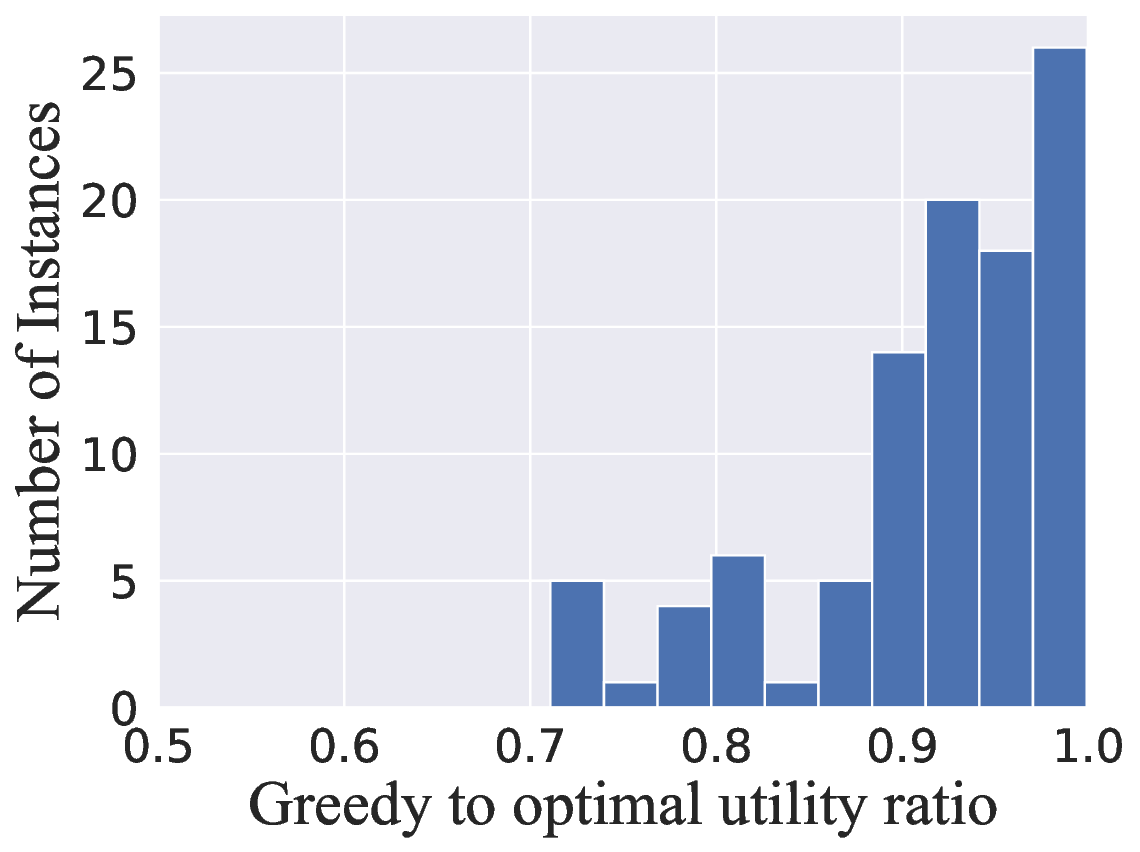}
        \caption{Relative performance ratio $\rho$ of the utility obtained by greedy placement to that of the optimal}
        \label{fig:greedy_opt_sense}
    \end{subfigure}
    \medskip
       \begin{subfigure}{0.3\textwidth}
    \centering
        \includegraphics[width=160pt]{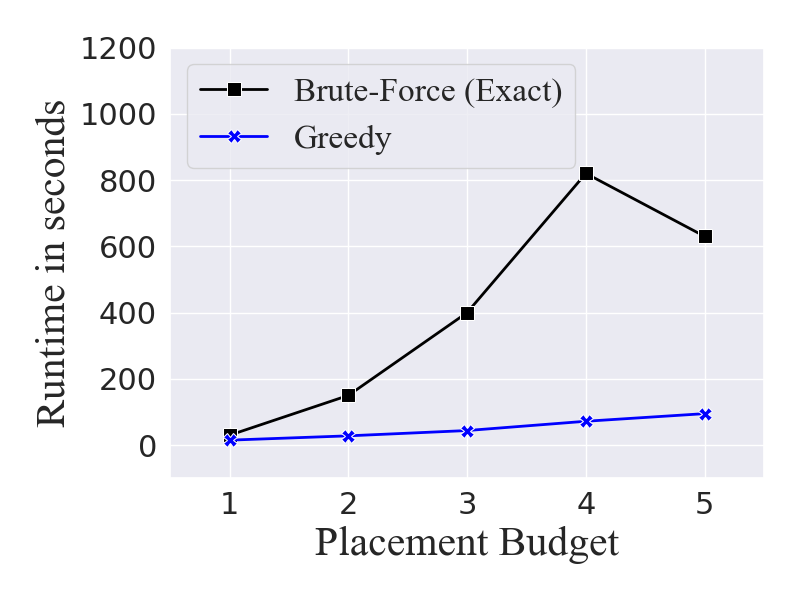}
        \caption{Runtime comparison of Algorithm~\ref{alg:sense} (greedy) and the brute-force (exact) oracle for varying placement budget $B$}
\label{fig:runtime_compare_sense}
    \end{subfigure}
    \hfill
            \begin{subfigure}{0.3\textwidth}
    \centering        \includegraphics[width=160pt]{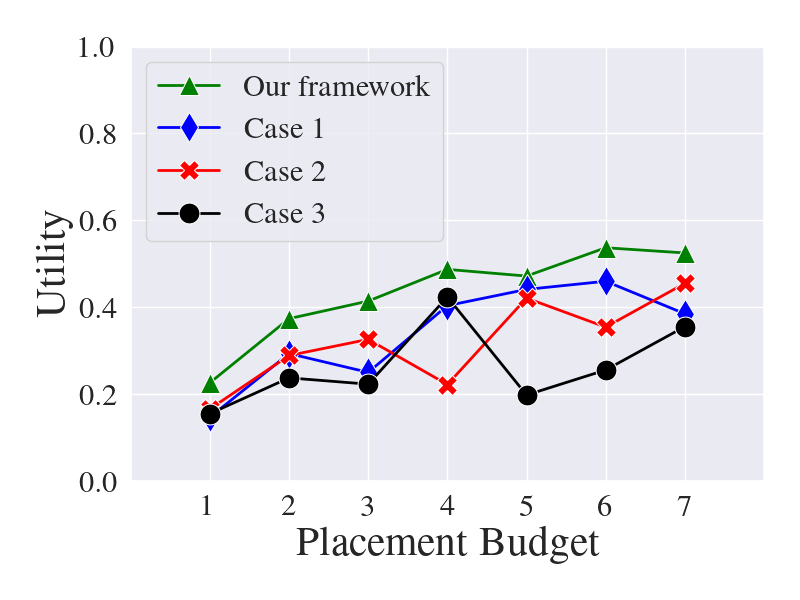}
        \caption{Comparison of utility of the proposed joint optimization framework with baseline cases}
        \label{fig:baseline_compare_joint}
    \end{subfigure}
    \hfill
    \begin{subfigure}{0.3\textwidth}
    \centering
        \includegraphics[width=155pt]{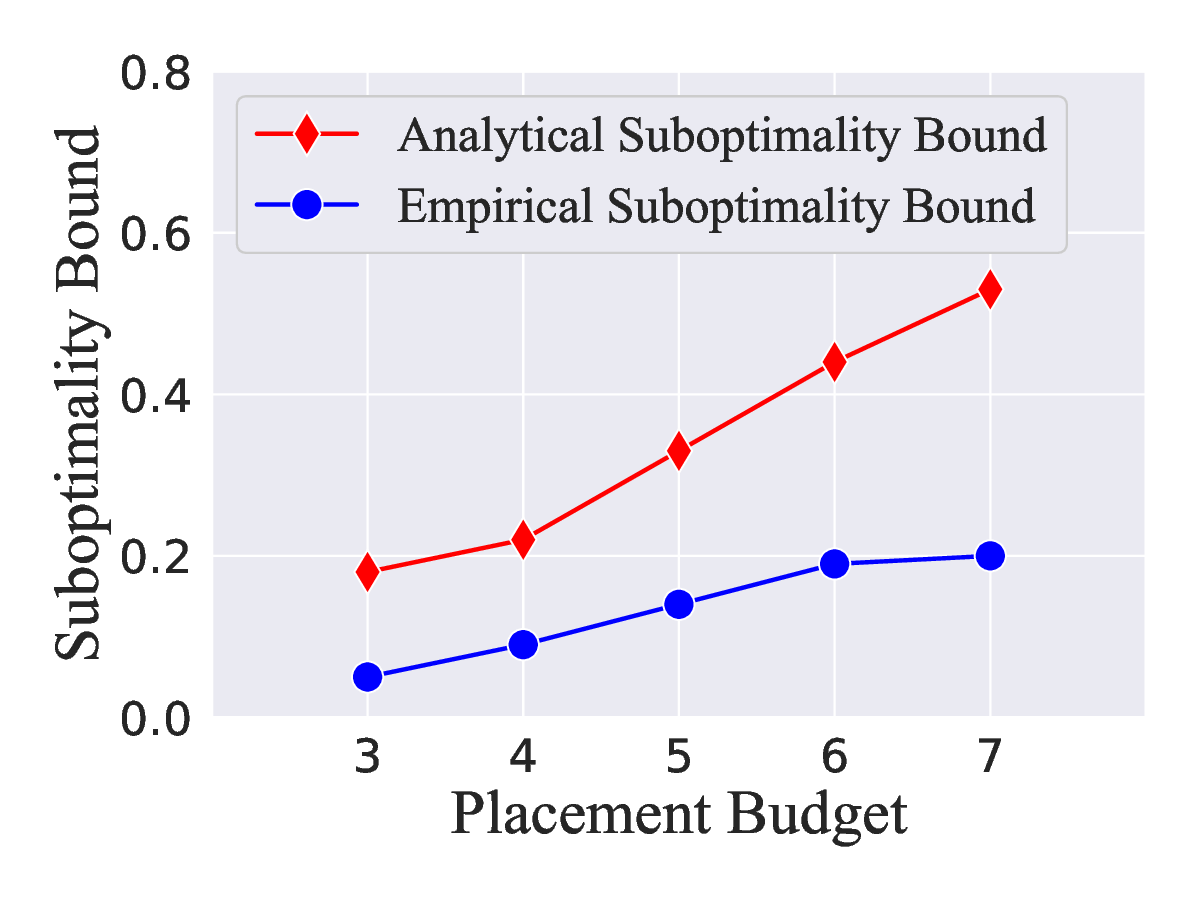}
        \caption{Comparision of empirical and analytical suboptimality bounds based on $\epsilon_{v, B}$  for varying budgets $B$.}
        \label{fig:subopt_bound_compare_sense}
    \end{subfigure}
    \caption{Empirical studies on Algorithm~\ref{alg:orient}, Algorithm~\ref{alg:sense}, and the joint optimization framework.}
    \label{fig:joint_sensor_opt_figures}
\end{figure*}
\subsection{Near-Optimality and Scalability of Algorithm~\ref{alg:sense}}

We empirically evaluate the  near-optimality of Algorithm~\ref{alg:sense}  by comparing its performance against the true optimal sensor placement on instances where exhaustive enumeration is computationally feasible. We consider a ground set of $|\mathcal{V}|=6$ sensors. We vary the placement budget $B \in \{1,2,3,4,5\}$. For each budget, we run Algorithm~\ref{alg:sense} to obtain the set $\cals^g$ and find the optimal placement $\sopt$  via exhaustive enumeration over all $\binom{|\mathcal{Q}|}{B}$ subsets of size $B$, for over $50$ instances. To ensure a fair comparison, all subsets, including the greedy solution, are evaluated using the same approximation quality of $\epsilon=0.001$ for the equilibrium value of their respective zero-sum games using Algorithm~\ref{alg:orient}. In Figure~\ref{fig:greedy_opt_sense}, we plot the utility ratio
\begin{equation}
\rho \;=\; 
\frac{v(\cals^g)}
     {v(\sopt)}.
\end{equation}
A ratio $\rho = 1$ indicates that the greedy solution is optimal, while values very close to $1$ indicate strong near-optimality. As observed from Figure~\ref{fig:greedy_opt_sense}, Algorithm~2 consistently achieves near-optimal performance across instances, with an average ratio of $0.92$. Moreover, in over $50\%$ of the instances, the greedy solution attains $\rho = 1$, exactly matching the optimal.

Next, we evaluate the computational efficiency of Algorithm~\ref{alg:sense} by comparing its runtime against the optimal oracle across varying selection budgets. The corresponding run times are reported in Figure~\ref{fig:runtime_compare_sense}. As observed, due to the number of possible subsets growing combinatorially (i.e., $\binom{|\mathcal{Q}|}{B}$), finding the exact (optimal) solution via exhaustive search (brute-force) becomes prohibitively expensive. In contrast, Algorithm~\ref{alg:sense} maintains tractable and scalable performance, exhibiting only modest growth in runtime as the budget increases.

 \subsection{Evaluation of the Joint Optimization Framework}

We aim to isolate and quantify the individual contributions of sensor placement and equilibrium-based scheduling optimizations in the proposed framework. Specifically, we evaluate the performance of the system when either the placement or the scheduling component is replaced by a non-adaptive (or de-coupled) baseline. We consider the following scenarios.

\begin{enumerate}
    \item \textit{Case 1 -- Greedy Placement with Uniform Scheduling}: 
    Sensor placement is performed using Algorithm~\ref{alg:sense}, but the scheduling policy is fixed to a uniform distribution over the four orientations for each deployed sensor. The utility used by Algorithm~\ref{alg:sense} (Line 3) under a rational intruder is
    \[
    v_{\scenario{uniform}}(\mathcal{S}) 
    = \min_{y \in Y} \bar{x}_{\mathcal{S}}^\top A^{\mathcal{S}} y,
    \]
    where $\bar{x}_{\mathcal{S}}$ denotes the uniform mixed strategy over joint orientations of the sensors in $\cals$.

    \item \textit{Case 2 -- Uniform Random Placement with Equilibrium-based Scheduling}: 
    A subset $\mathcal{S}^{\scenario{uniform}}$ of size $B$ is sampled uniformly at random from $\mathcal{Q}$, and the induced game is solved using Algorithm~\ref{alg:orient} to find the equilibrium value
    \[
    v(\mathcal{S}^{\scenario{uniform}}) 
    = \max_{x \in X} \min_{y \in Y} x^\top A^{\mathcal{S}^{\scenario{uniform}}} y.
    \]

    \item \textit{Case 3 -- Uniform Random Placement with Uniform Scheduling}: 
    Both placement and scheduling are randomized. A subset $\mathcal{S}^{\scenario{uniform}}$ of size $B$ is sampled uniformly at random from $\mathcal{Q}$, and scheduling is fixed to the uniform orientation distribution $\bar{x}_{\mathcal{S}^{\scenario{uniform}}}$. The performance is evaluated as
    \[
    v_{\scenario{uniform}}(\mathcal{S}^{\scenario{uniform}}) = \min_{y \in Y} \bar{x}_{\mathcal{S}^{\scenario{uniform}}}^\top A^{\mathcal{S}^{\scenario{uniform}}} y.
    \]

    \item \textit{\textbf{Case 4 -- Our Framework}}: We run greedy placement (Algorithm~2) using the equilibrium-based scheduling (Algorithm~1) as the oracle to determine the placement $\cals^g$ and the resulting value is $v(\cals^g)$.
\end{enumerate}

We set $|\mathcal{V}| = 8$ and randomly generate the sensor detection probabilities (as mentioned earlier) and fix (freeze) them. We vary the placement budget $B \in \{1,2,3,4,5,6,7\}$. Then, for each of the aforementioned cases,  we compute the overall utility using a fixed approximation tolerance of $\epsilon = 0.001$.

\begin{remark}
This experiment is designed to disentangle and quantify the individual and joint contributions of strategic placement and equilibrium scheduling. Comparing Case 1 with our framework isolates the marginal benefit of equilibrium scheduling when placement is already optimized. Comparing Case 2 with our framework isolates the value of strategic placement under optimal scheduling. Finally, Case 3 provides a fully non-adaptive baseline, allowing us to clearly measure the compounded gains obtained by jointly optimizing both components. Together, these comparisons offer a clean and systematic decomposition of where performance improvements originate.
\end{remark}

In Figure~\ref{fig:baseline_compare_joint}, we compare the achieved detection utilities across the four cases. As observed, the joint optimization framework (Case 4) attains the highest utility. This clearly demonstrates that sensor placement and scheduling are inherently coupled decisions. Optimizing them independently leads to suboptimal performance, whereas jointly optimizing both components is critical and beneficial for maximizing worst-case detection performance against a rational intruder.

Finally, in Figure~\ref{fig:subopt_bound_compare_sense}, we compare the analytical suboptimality bounds for Algorithm~\ref{alg:sense}, derived via the weak-submodularity constant $\epsilon_{v,B}$, with empirically estimated bounds across varying budgets $B$. We observe that the bounds from Lemma~\ref{lma:weaksubmod} are indeed conservative, as the empirically estimated instance dependent bounds are significantly tighter.
\section{Conclusion}
We addressed the problem of joint sensor placement and orientation scheduling for an intrusion detection task by decomposing the overall task into a meta-level subset selection problem for placement combined with a game-theoretic orientation scheduling subproblem. For a fixed placement, the scheduling component is formulated as a zero-sum game between the defender and a rational intruder, whose Nash equilibrium value characterizes the best achievable detection performance. We developed an efficient equilibrium-seeking algorithm with convergence guarantees and used the resulting game value as the utility for the placement meta-problem. We further showed that this utility is weakly submodular over the set of sensors, which enables greedy selection with provable near-optimality guarantees and yields a computationally tractable and theoretically grounded solution to the joint optimization problem. Through extensive simulations, we demonstrated that the proposed framework achieves superior  detection performance while significantly improving computational efficiency relative to baseline approaches.
\appendix

\section{Proof of Theorem~\ref{thm:dist_best_response}}
\begin{proof}
Fix an intruder mixed strategy $y \in Y$, and let $\cals=\{s_1,\dots,s_{|\cals|}\}$ be the selected sensor set.  
Recall that the defender's \emph{joint} pure strategy space is $
I \;=\; \Theta^{\cals}$. Each strategy can be represented as
$i = [i_1,\dots,i_{|\cals|}] \in I, 
m \;=\; |I| =  |\Theta|^{|\cals|}.$
Here each coordinate $i_q\in i$ specifies the orientation chosen for sensor $s_q$. We use the notation $\theta_k \in \Theta$,  $k = 1,\hdots, |\Theta|$, to denote the pure (orientation) strategy for a sensor. 

For each sensor $s_q\in\cals$, define we have its sub-game matrix $A^q \in \mathbb{R}^{|\Theta|\times |Y|}$ given by
\[
A^q_{k j} \;=\; -\, V^q_{k j}\,\log\!\bigl(1-p_{detect,q}\bigr),
\quad \theta_k\in\Theta,\; j\in J.
\]
For a joint pure strategy $i=[i_1,\dots,i_{|\cals|}]$, the entries of the \emph{overall} game matrix
$A^{\cals}\in\mathbb{R}^{m\times |Y|}$ can be equivalently written as the component-wise sum of entries of the sub-game matrices, as. 
\begin{equation}\label{eq:additive-lift}
A^{\cals}_{i j} \;=\; \sum_{q=1}^{|\cals|} A^q_{k j},
\quad \forall \hspace{2pt} i \in I,  \theta_k \in\Theta,\; j\in J.
\end{equation}

This follows from \eqref{eq:log_p_miss} and \eqref{eq:small_A}. For each sensor $s_q$, let
\[
\sigma^{q\star} \in \argmin_{\sigma^q\in\Delta_{|\Theta|}} \sigma^{q\top} A^q y
\]
be a (mixed) best response to $y$ in its sub-game. Define the defender's joint mixed strategy
\[
x^\star \;=\; \bigotimes_{q=1}^{|\cals|} \sigma^{q\star} \in X,
\]
i.e., under $x^\star$ the orientations are sampled independently with marginals $\sigma^{q\star}$.

\smallskip
\noindent\textbf{Marginalization of the joint strategy.}
Let $x\in X$ be any joint mixed strategy over $\Theta^{|\cals|}$. For each sensor index $q$,
let $\mu_q\in\Delta_{|\Theta|}$ denote the marginal distribution of $x$ on coordinate $i_q$, i.e.,
\[
\mu_q(\theta_k) \;=\; \sum_{i\in\Theta^{|\cals|}: \, i_q=\theta_k} x(i), 
\qquad \forall \theta_k\in\Theta.
\]
Using \eqref{eq:additive-lift} and applying linearity of expectation,
\begin{align*}
x^\top A^{\cals} y
&= \sum_{i\in\Theta^{|\cals|}} x(i)\,\bigl(A^{\cals} y\bigr)
 = \sum_{i} x(i)\, \sum_{q=1}^{|\cals|} \bigl(A^q y\bigr) \\
&= \sum_{q=1}^{|\cals|} \sum_{i} x(i)\, \bigl(A^q y\bigr)
 = \sum_{q=1}^{|\cals|} \sum_{\theta_k\in\Theta} \Bigl(\sum_{i:\, i_q=\theta_k} x(i)\Bigr)\, \bigl(A^q y\bigr) \\
&= \sum_{q=1}^{|\cals|} \sum_{\theta_k\in\Theta} \mu_q(\theta_k)\, \bigl(A^q y\bigr)
 = \sum_{q=1}^{|\cals|} \mu_q^\top A^q y.
\end{align*}

\smallskip
\noindent\textbf{Minimization separates across sensors.}
Since for any choice of marginals $(\mu_1,\dots,\mu_{|\cals|})$ there exists a feasible joint distribution
(e.g., the product $\bigotimes_q \mu_q$), we have
\begin{align*}
\min_{x\in\Delta_m} x^\top A^{\cals} y
&= \min_{\mu_1,\dots,\mu_{|\cals|}\in\Delta_{|\Theta|}} \ \sum_{q=1}^{|\cals|} \mu_q^\top A^q y\\
 &= \sum_{q=1}^{|\cals|} \min_{\mu_q\in\Delta_{|\Theta|}} \mu_q^\top A^q y.
\end{align*}
By definition of $\sigma^{q\star}$, each $\sigma^{q\star}$ attains the corresponding minimum, hence
\[
\sum_{q=1}^{|\cals|} \sigma^{q\star\top} A^q y
\;=\;
\sum_{q=1}^{|\cals|} \min_{\mu_q\in\Delta_{|\Theta|}} \mu_q^\top A^q y
\;=\;
\min_{x\in\Delta_m} x^\top A^{\cals} y.
\]

\smallskip
\noindent\textbf{The product strategy is a best-response.}
The product strategy $x^\star=\bigotimes_q \sigma^{q\star}$ has marginals $\mu_q=\sigma^{q\star}$ for all $q$,
so by Step 1,
\[
x^{\star\top} A^{\cals} y
\;=\;
\sum_{q=1}^{|\cals|} \sigma^{q\star\top} A^q y
\;=\;
\min_{x\in X} x^\top A^{\cals} y.
\]
Therefore,
\[
x^\star \in \argmin_{x\in X} x^\top A^{\cals} y.
\]\end{proof}

\section{Proof of Theorem~\ref{thm:nashgap}}
\begin{proof}
From Theorem~\ref{thm:dist_best_response}, we have established that the product strategy $x^\star=\bigotimes_q \sigma^{q\star}$ where $\sigma^{q\star} \in \argmin_{\sigma^q\in\Delta_{|\Theta|}} \sigma^{q\top} A^q y$ is the best-response to the intruder's mixed-strategy $y$ with respect to the overall game matrix $A^\cals$. Following the discussions in Sections~2.4, 2.5 and 2.6 of \cite{freund1996game}, the resulting repeated distributed best-response dynamics presented in Algorithm~\ref{alg:orient} can be embedded within the multiplicative-weight update framework presented in \cite{freund1996game}. Invoking Corollary~2 of \cite{freund1996game}, we obtain convergence of the empirical strategy profile to an approximate Nash equilibrium, with the standard no-regret guarantees and rate bounds established therein.
\end{proof}

\section{Proof of Lemma~\ref{lma:weaksubmod}}

We begin by constructing a projected game representation that allows sensor addition to be interpreted directly as a matrix addition. 
Specifically, we embed all game matrices arising from sensor sets of varying cardinalities into a common space.
Under this construction, the game matrix corresponding to any sensor set $\mathcal S$ can be written as a linear sum of per-sensor matrices associated with the elements of $\mathcal S$. 
Consequently, adding a sensor corresponds to adding its associated matrix contribution in this unified representation.

We show that this projected representation preserves both mixed strategies and the equilibrium value. In particular, there exist projection and lifting operators between the original game (which we call the ``native" game)  and the projected game such that the projected mixed-strategies combined with the projected game matrix lead to the same game value, effectively.

Recall that $\mathcal V$ is the full candidate sensor set and 
$\Theta$ is the finite set of orientations available to each sensor. Each element of $\mathcal V$ is interpreted as a sensor–location pair, and we suppress any explicit $\mathcal V \times \mathcal L$ structure, since the induced game matrix depends only on the realized placement configuration. So we have that $|\mathcal{V}|$ is the size of the largest possible sensor placement set. For notational uniformity across different placement sets,
we define the largest joint orientation space
\[
\Theta^{\mathcal V}
\;:=\;
\prod_{s\in\mathcal V} \Theta,
\]
and let $\alpha = (\alpha_s)_{s\in\mathcal V} \in \Theta^{\mathcal V}$ 
denote a full orientation assignment.

\noindent \textbf{Projection of the  game matrix.}
Define the matrix 
\[
\mathbf A^{\mathcal V} 
\;\in\;
\mathbb R^{|\Theta|^{|\mathcal V|} \times |J|}
\]
with rows indexed by $\alpha \in \Theta^{\mathcal V}$ 
and columns by $j \in  J$, and entries
\[
\mathbf A^{\mathcal V}(\alpha,j)
\;:=\;
\sum_{s\in\mathcal V} a_s(\alpha_s,j),
\]
where $a_s(\theta,j)$ denotes the payoff contribution of sensor $s$
when oriented at $\theta \in \Theta$ against intruder path $j \in J$, which is given by

\[
a_s(\theta,j) = -V^s_{\theta j} \log(1 - p_{detect,s}).
\]

Equivalently, we can decompose the matrix as
\[
\mathbf A^{\mathcal V}
\;=\;
\sum_{s\in\mathcal V} \mathbf A_s,
\qquad
\mathbf A_s(\alpha,j)
:=
a_s(\alpha_s,j).
\]
Each $\mathbf A_s$ depends only on coordinate $\alpha_s$,
and hence has repeated entries along all other coordinates.

\noindent \textbf{Projected game matrix of a sensor set $\cals$. }
For any subset $\mathcal S \subseteq \mathcal V$,
define the projected matrix
\[
\mathbf A^{\mathcal S}
\;\in\;
\mathbb R^{|\Theta|^{|\mathcal V|} \times | J|}
\]
by
\[
\mathbf A^{\mathcal S}(\alpha,j)
\;:=\;
\sum_{s\in\mathcal S} a_s(\alpha_s,j)
\;=\;
\sum_{s\in\mathcal S} \mathbf A_s(\alpha,j).
\]
Thus $\mathbf A^{\mathcal S}$ lives in the same space
as $\mathbf A^{\mathcal V}$, but ignores all coordinates
$s \notin \mathcal S$.
In particular, for $s \notin \mathcal S$,
$\mathbf A^{\mathcal S}(\alpha,j)$ is constant in $\alpha_s$,
which induces repeated rows along those coordinates.

Let the native (i.e., non-projected) matrix for subset $\mathcal S$
be denoted by
\[
\widetilde{\mathbf A}^{\mathcal S}
\;\in\;
\mathbb R^{|\Theta|^{|\mathcal S|} \times | J|},
\qquad
\widetilde{\mathbf A}^{\mathcal S}(\beta,j)
:=
\sum_{s\in\mathcal S} a_s(\beta_s,j),
\]
with $\beta \in \Theta^{\mathcal S}$.

\noindent \textbf{Projection (marginalization) of defender strategies.}
Let $\Delta(X)$ denote the probability simplex over a finite set $X$.
For $x \in \Delta(\Theta^{\mathcal V})$,
define the marginal distribution on $\mathcal S$ as
\[
\pi_{\mathcal S}(x) \in \Delta(\Theta^{\mathcal S}),
\qquad
\big(\pi_{\mathcal S}(x)\big)(\beta)
:=
\sum_{\alpha \in \Theta^{\mathcal V} :
\alpha|_{\mathcal S} = \beta} x(\alpha).
\]

\medskip
\noindent
\textbf{Identity 1 - Mixed-strategy equivalence.}
For every $x \in \Delta(\Theta^{\mathcal V})$
and every $j \in  J$,
\begin{align*}
\sum_{\alpha \in \Theta^{\mathcal V}}
x(\alpha)\mathbf A^{\mathcal S}(\alpha,j)
&=
\sum_{\beta \in \Theta^{\mathcal S}}
\big(\pi_{\mathcal S}(x)\big)(\beta)
\widetilde{\mathbf A}^{\mathcal S}(\beta,j).
\end{align*}
This can be easily established as follows. Since $\mathbf A^{\mathcal S}(\alpha,j)$
depends on $\alpha$ only through $\alpha|_{\mathcal S}$,
grouping terms by $\beta = \alpha|_{\mathcal S}$ yields
\begin{align*}
\sum_{\alpha} x(\alpha)\mathbf A^{\mathcal S}(\alpha,j)
&=
\sum_{\beta \in \Theta^{\mathcal S}}
\left(
\sum_{\alpha:\alpha|_{\mathcal S}=\beta}
x(\alpha)
\right)
\sum_{s\in\mathcal S} a_s(\beta_s,j).
\end{align*}

Conversely, for any $\bar x \in \Delta(\Theta^{\mathcal S})$,
define a lifted distribution
\[
\mathrm{Lift}_{\mathcal S}(\bar x)(\alpha)
:=
\bar x(\alpha|_{\mathcal S})
\prod_{s \notin \mathcal S} q_s(\alpha_s),
\]
where each $q_s \in \Delta(\Theta)$ is arbitrary
(e.g., uniform).
Then
\[
\pi_{\mathcal S}(\mathrm{Lift}_{\mathcal S}(\bar x))
=
\bar x,
\]
and the induced payoff against any $j$
coincides with that of $\bar x$ in the native game.

\noindent \textbf{Identity 2 - Value equivalence.}
Define
\[
v_{\mathcal S}(x,y)
:=
x^\top \mathbf A^{\mathcal S} y,
\qquad
x \in \Delta(\Theta^{\mathcal V}),\;
y \in \Delta( J),
\]
and
\[
\tilde v_{\mathcal S}(\bar x,y)
:=
\bar x^\top \widetilde{\mathbf A}^{\mathcal S} y,
\qquad
\bar x \in \Delta(\Theta^{\mathcal S}).
\]
Then
\[
v_{\mathcal S}(x,y)
=
\tilde v_{\mathcal S}(\pi_{\mathcal S}(x),y),
\]
and for every $\bar x$ there exists
$x = \mathrm{Lift}_{\mathcal S}(\bar x)$
such that
\[
\tilde v_{\mathcal S}(\bar x,y)
=
v_{\mathcal S}(x,y).
\]
Therefore,
\[
\max_{x \in \Delta(\Theta^{\mathcal V})}
\min_{y \in \Delta(\mathcal J)}
v_{\mathcal S}(x,y)
=
\max_{\bar x \in \Delta(\Theta^{\mathcal S})}
\min_{y \in \Delta(\mathcal J)}
\tilde v_{\mathcal S}(\bar x,y),
\]
and equilibrium strategies in the  projected game
can be projected onto $\mathcal S$
without changing the game value.\\

\noindent \textbf{Example to Illustrate the Projection --} We now present an example to illustrate the projected representation of the game matrix, their mixed-strategies and the equivalence of the game value. Let $\mathcal V=\{1,2\}$ and $\Theta=\{a,b\}$, so the projected defender action space is
\[
\Theta^{\mathcal V}=\{(a,a),(a,b),(b,a),(b,b)\}.
\]
Let $\mathcal S=\{1\}$ be a single-sensor set. The \emph{native} game for $\mathcal S$ has defender actions
\[
\Theta^{\mathcal S}=\{a,b\}.
\]
Fix any intruder action set $ J$ (e.g., paths). Let sensor~1's payoff contributions be
$a_1(a,j)$ and $a_1(b,j)$ for each $j\in J$.
In the projected representation, the (projected) game matrix $\mathbf A^{\mathcal S}$ is indexed by
$\alpha=(\alpha_1,\alpha_2)\in\Theta^{\mathcal V}$ and satisfies
\[
\mathbf A^{\mathcal S}(\alpha,j)=a_1(\alpha_1,j),
\qquad \forall \alpha\in\Theta^{\mathcal V},\ \forall j\in J,
\]
so the second coordinate $\alpha_2$ is irrelevant and induces repeated rows. To make this explicit, write the rows in the order
$(a,a),(a,b),(b,a),(b,b)$. Then for each fixed $j\in J$,
\[
\mathbf A^{\mathcal S}(:,j)
=
\begin{bmatrix}
a_1(a,j)\\
a_1(a,j)\\
a_1(b,j)\\
a_1(b,j)
\end{bmatrix},
\]
where the first two entries coincide (sensor~1 plays $a$) and the last two entries coincide (sensor~1 plays $b$).

\noindent \textit{Projection of mixed strategies.}
A mixed strategy in the projected game is any distribution
$x\in\Delta(\Theta^{\mathcal V})$ over the four joint actions:
\begin{align*}
x_{aa}=x(a,a),&\quad x_{ab}=x(a,b),\\
x_{ba}=x(b,a),&\quad x_{bb}=x(b,b).
\end{align*}
We have $x_{aa}+x_{ab}+x_{ba}+x_{bb}=1.$
The projection (marginal) of $x$ onto $\mathcal S=\{1\}$ is the distribution
$\pi_{\mathcal S}(x)\in\Delta(\Theta^{\mathcal S})$ given by
\[
\pi_{\mathcal S}(x)(a)=x_{aa}+x_{ab},
\qquad
\pi_{\mathcal S}(x)(b)=x_{ba}+x_{bb}.
\]
Because $\mathbf A^{\mathcal S}(\alpha,j)$ depends only on $\alpha_1$, the expected payoff against any
intruder mixed strategy $y\in\Delta( J)$ depends only on this marginal:
\begin{align*}
&x^\top \mathbf A^{\mathcal S} y
=
\sum_{j\in\mathcal J} y(j)\sum_{\alpha\in\Theta^{\mathcal V}} x(\alpha)\mathbf A^{\mathcal S}(\alpha,j)\\
&=
\sum_{j\in\mathcal J} y(j)\Big[(x_{aa}+x_{ab})a_1(a,j) + (x_{ba}+x_{bb})a_1(b,j)\Big]\\
&=
\sum_{j\in\mathcal J} y(j)\Big[\pi_{\mathcal S}(x)(a)\,a_1(a,j)+\pi_{\mathcal S}(x)(b)\,a_1(b,j)\Big]\\
&=
\big(\pi_{\mathcal S}(x)\big)^\top \widetilde{\mathbf A}^{\mathcal S} y,
\end{align*}
where $\widetilde{\mathbf A}^{\mathcal S}$ is the native $2\times| J|$ matrix with rows indexed by
$\{a,b\}$ and entries $\widetilde{\mathbf A}^{\mathcal S}(a,j)=a_1(a,j)$ and
$\widetilde{\mathbf A}^{\mathcal S}(b,j)=a_1(b,j)$.

\noindent \textit{Lifting (one concrete choice).}
Conversely, given any native mixed strategy $\bar x\in\Delta(\Theta^{\mathcal S})$ with
$\bar x(a)=p$ and $\bar x(b)=1-p$, one valid lift to the projected space is to distribute mass
uniformly over the irrelevant coordinate:
\[
x_{aa}=x_{ab}=\frac{p}{2},
\qquad
x_{ba}=x_{bb}=\frac{1-p}{2}.
\]
This lifted $x$ satisfies $\pi_{\mathcal S}(x)=\bar x$ and induces the same expected payoff against all $y$,
illustrating that the projection preserves mixed strategies (up to marginalization) and hence preserves the game value.\\

\noindent \textbf{Identity 3 - Non-decreasing max-min value. }Let $\mathbf M,\mathbf M'\in\mathbb{R}^{m\times n}$ satisfy $\mathbf M'\succeq \mathbf M$ entry wise. Write $\mathbf M'=\mathbf M+\Delta\mathbf M$, where $\Delta\mathbf M\succeq 0$ entry wise.
Let $x^\star\in\arg\max_{x\in\Delta([m])}\min_{y\in\Delta([n])} x^\top \mathbf M y$ be an optimal strategy for $\mathbf M$.
Define
\[
\underline v := \min_{y\in\Delta([n])} (x^\star)^\top \mathbf M y
\;=\;
v(\mathbf M).
\]
Now consider the same mixed strategy $x^\star$ evaluated on $\mathbf M'$:
\[
\min_{y\in\Delta([n])} (x^\star)^\top \mathbf M' y
=
\min_{y\in\Delta([n])}\Big((x^\star)^\top \mathbf M y + (x^\star)^\top \Delta\mathbf M\, y\Big).
\]
Since $\Delta\mathbf M\succeq 0$ entry wise and $x^\star,y$ are probability vectors (hence entry wise nonnegative),
we have $(x^\star)^\top \Delta\mathbf M\, y \ge 0$ for every $y\in\Delta([n])$. Therefore,
\[
(x^\star)^\top \mathbf M y + (x^\star)^\top \Delta\mathbf M\, y
\;\ge\;
(x^\star)^\top \mathbf M y,
\qquad \forall y\in\Delta([n]).
\]
Taking the minimum over $y$ preserves this point wise inequality, yielding
\[
\min_{y\in\Delta([n])} (x^\star)^\top \mathbf M' y
\;\ge\;
\min_{y\in\Delta([n])} (x^\star)^\top \mathbf M y
\;=\;
v(\mathbf M).
\]
Finally, maximizing over $x$ on the left can only increase the value:
\begin{align*}
v(\mathbf M')
&=
\max_{x\in\Delta([m])}\min_{y\in\Delta([n])} x^\top \mathbf M' y \\
\;&\ge\;
\min_{y\in\Delta([n])} (x^\star)^\top \mathbf M' y
\;\ge\;
v(\mathbf M).
\end{align*}

We are now in place to establish the monotonicity and weak-submodularity properties of $v$ (Lemma~\ref{lma:weaksubmod}).

\begin{proof} ~\\

\noindent \textbf{Proof of Monotonicity. } For each sensor $s\in\mathcal V$, we have its \emph{per-sensor projected (projected) matrix}
\[
\mathbf A_s \in \mathbb{R}^{|\Theta|^{|\mathcal V|}\times | J|},
\qquad
\mathbf A_s(\alpha,j) := a_s(\alpha_s,j),
\]
where $\alpha\in\Theta^{\mathcal V}$ and $j\in J$. Let the value of the game induced by $\mathcal S$ be
\begin{equation}
\label{eq:value_def}
v(\mathcal S)
\;:=\;
\max_{x\in\Delta(\Theta^{\mathcal V})}\;
\min_{y\in\Delta( J)}
x^\top \mathbf A^{\mathcal S} y.
\end{equation}
Equivalently, this equals the native game value over $\Theta^{\mathcal S}$ (from \textbf{Identity 2}).

We have that each sensor's contribution is nonnegative:
\begin{equation}
\label{eq:nonneg_assump}
a_s(\theta,j)\ge 0,
\quad \forall s\in\mathcal V,\ \forall \theta\in\Theta,\ \forall j\in J.
\end{equation}
Under \eqref{eq:nonneg_assump}, each $\mathbf A_s$ is entry wise nonnegative, i.e.,
$\mathbf A_s(\alpha,j)\ge 0$ for all $(\alpha,j)$.

Let $\mathcal S\subseteq\mathcal T$. Using the per-sensor decomposition,
\[
\mathbf A^{\mathcal T}
=
\sum_{s\in\mathcal T}\mathbf A_s
=
\sum_{s\in\mathcal S}\mathbf A_s
+
\sum_{s\in\mathcal T\setminus\mathcal S}\mathbf A_s
=
\mathbf A^{\mathcal S} + \sum_{s\in\mathcal T\setminus\mathcal S}\mathbf A_s.
\]
By \eqref{eq:nonneg_assump}, each $\mathbf A_s\succeq 0$ entry wise, hence
\[
\mathbf A^{\mathcal T}\succeq \mathbf A^{\mathcal S}\quad \text{entry wise}.
\]
Applying \textbf{Identity 3} with $\mathbf M=\mathbf A^{\mathcal S}$ and
$\mathbf M'=\mathbf A^{\mathcal T}$, and recalling the value definition \eqref{eq:value_def},
we obtain $v(\mathcal T)\ge v(\mathcal S)$, which proves monotonicity.

\noindent \textbf{Proof of Weak-Submodularity.}  With a slight abuse of notation, we use $v(\mathcal S)$ and 
$v(\mathbf A^{\mathcal S})$ interchangeably to denote the optimal value 
of the zero-sum game induced by the sensor set $\mathcal S$. 
That is,
\[
v(\mathcal S)
\;=\;
v(\mathbf A^{\mathcal S})
\;:=\;
\max_{x\in\Delta(\Theta^{\mathcal V})}
\min_{y\in\Delta( J)}
x^\top \mathbf A^{\mathcal S} y.
\]

Let $\|\cdot\|_\infty$ denote the entry wise max norm. The maximum per-sensor contribution magnitude $\xi(s)$, which is defined in \eqref{eq:lambda_bound}, is exactly the largest entry of the sensor matrix. That is
\[
\xi(s)\;:=\;\|\mathbf A_s\|_\infty
\;=\;
\max_{\alpha\in\Theta^{\mathcal V},\, j\in J}\; \mathbf A_s(\alpha,j),
\qquad s\in\mathcal V.
\]

\noindent \textbf{Identity 4 - Sensitivity of game-value to perturbation~\cite{lipton2006stability}. }For two zero-sum games $\mathbf M, \mathbf M'\in\mathbb R^{m\times n}$,
\begin{align*}
&\Big|v(\mathbf M)-v(\mathbf M')\Big|
\;\le\;
\|\mathbf M-\mathbf M'\|_\infty.
\end{align*}

 Now consider the difference in marginal gains.
\begin{align*}
&v_i(\mathcal T)-v_i(\mathcal S)
=
\big(v(\mathcal T\cup\{i\})-v(\mathcal T)\big)
-
\big(v(\mathcal S\cup\{i\})-v(\mathcal S)\big)\\
&=
\Big(v(\mathbf A^{\mathcal T}+\mathbf A_i)-v(\mathbf A^{\mathcal T})\Big)
-
\Big(v(\mathbf A^{\mathcal S}+\mathbf A_i)-v(\mathbf A^{\mathcal S})\Big)\\
&=
\Big(v(\mathbf A^{\mathcal T}+\mathbf A_i)-v(\mathbf A^{\mathcal S}+\mathbf A_i)\Big)
+
\Big(v(\mathbf A^{\mathcal S})-v(\mathbf A^{\mathcal T})\Big).
\end{align*}

Applying \textbf{Identity 4}, we have
\begin{align*}
v_i(\mathcal T)-v_i(\mathcal S)
\;& \le\;
\|\mathbf A^{\mathcal T}-\mathbf A^{\mathcal S}\|_\infty
+
\|\mathbf A^{\mathcal T}-\mathbf A^{\mathcal S}\|_\infty \\
&=
2\|\mathbf A^{\mathcal T}-\mathbf A^{\mathcal S}\|_\infty.
\end{align*}
Finally, since $\mathbf A^{\mathcal T}-\mathbf A^{\mathcal S}
=
\sum_{s\in\mathcal T\setminus\mathcal S}\mathbf A_s$, we have 
\[
\|\mathbf A^{\mathcal T}-\mathbf A^{\mathcal S}\|_\infty
\le
\sum_{s\in\mathcal T\setminus\mathcal S}\|\mathbf A_s\|_\infty
=
\sum_{s\in\mathcal T\setminus\mathcal S}\xi(s).
\]

For a fixed $\mathcal T$, the right-hand side is maximized over $\mathcal S\subseteq\mathcal T$
by taking $\mathcal S=\emptyset$, yielding $2\sum_{s\in\mathcal T}\xi(s)$.
Next, over all $\mathcal T$ with $|\mathcal T|\le B$, the quantity $\sum_{s\in\mathcal T}\xi(s)$
is maximized by choosing $\mathcal T=\mathcal S^B$, i.e., the $B$ sensors with the largest $\xi(s)$.
Therefore,
\[
\epsilon_{f,B}
=
\max_{\substack{\mathcal S\subseteq\mathcal T\subseteq\mathcal V\\ |\mathcal T|\le B}}
\ \max_{i\in\mathcal V\setminus\mathcal T}
\ \big(v_i(\mathcal T)-v_i(\mathcal S)\big)
\;\le\;
2\sum_{s\in\mathcal S^B}\xi(s).
\]
\end{proof}

\section{Proof of Theorem~\ref{thm:greedy}}
\begin{proof}Lemma~\ref{lma:weaksubmod} shows that \sense~can be seen as a monotone weak-submodular function maximization problem under a partition matroid constraint. Consequently, the near-optimality guarantees for Algorithm~\ref{alg:sense} follow directly from the standard results for weak-submodular maximization presented in~\cite{fisher2009analysis,calinescu2011maximizing,hashemi2019submodular}. Specifically, greedy algorithm has the following near-optimality bound for maximizing a set function $f$ with additive weak-submodularity constant $\epsilon_f$ under $p-$matroid constraints and a cardinality constraint of $k$.

\[
f(\cals^g) \ge \frac{p}{p+1} (f(\sopt) - (k-1) \epsilon_f),
\]
where $\cals^g$ and $\sopt$ are the greedy and optimal sets, respectively.
By substituting $p = 1$, $k = B$ and $\epsilon_f = \epsilon_{v,B}$, we obtain the stated result. \end{proof}

\section{Proof of Corollary~\ref{coro:same_sensors}}

\begin{proof}
    If all sensors are identical (or homogenous), i.e., they have the same $p_{detect}$, orientation action space $\Theta$,  and have the same coverage/range (i.e., same $V_{ij}$'s),  \sense~problem reduces to the problem of selecting the best subset of $B$ locations (or graph nodes) to place sensors that maximizes the utility. Then, this becomes equivalent to a weak-submodular maximization problem under (only) cardinality constraints. The result follows directly from Proposition~2 of \cite{hashemi2019submodular}.
\end{proof}
\bibliographystyle{ieeetr}
\bibliography{ref}

\end{document}